\documentclass[a4paper,english]{lipics-v2018}
\usepackage[utf8]{inputenc}
\usepackage{ifthen}
\usepackage{amssymb}
\usepackage{amsmath}
\usepackage{amsthm}
\usepackage{mathtools}
\usepackage{color}
\usepackage{hyperref}
\usepackage{comment}
\usepackage{stmaryrd}
\usepackage[inline]{enumitem}
\usepackage{wrapfig}
\newlist{enuminline}{enumerate*}{1}
\setlist[enuminline]{label=(\arabic*)}
\setlist[itemize]{noitemsep, topsep=0pt}
\setlist[enumerate]{noitemsep, topsep=0pt}
\newboolean{short} %declaration
\setboolean{short}{true} %assignment
\setboolean{short}{false} %assignment
\newcommand{\nats}{\mathbb{N}}

\DeclarePairedDelimiter\floor{\lfloor}{\rfloor}
\newcommand{\pfun}{\rightharpoonup}
\newcommand{\card}[1]{\left|#1\right|}
\newcommand{\len}[1]{\left|#1\right|}
\newcommand{\set}[1]{\left\{#1\right\}}

\newcommand{\graph}[1]{\Gamma(#1)}

\newcommand{\ran}{{\mathit{ran}}}
\newcommand{\truthof}[1]{{\llbracket}#1{\rrbracket}}

\newcommand{\mcinit}{\iota}
\newcommand{\norm}{\mathit{norm}}

\newcommand{\evstrong}{\diamond\mathcal{S}}
\newcommand{\evweak}{\diamond\mathcal{W}}

\newcommand{\alg}{\mathcal{A}}
\newcommand{\fds}{\mathit{FD}}
\newcommand{\procs}{\Pi}

\newcommand{\fdran}{\mathcal{R}}
\newcommand{\sspace}{\Sigma}
\newcommand{\sspacep}{{\Sigma_p}}
\newcommand{\sspaceq}{{\Sigma_q}}
\newcommand{\msgs}{\mathcal{M}}
\newcommand{\vals}{\mathcal{V}}
\newcommand{\hbeat}{\mathit{hb}}
\newcommand{\css}{\mathit{CSS}}
\newcommand{\nnext}{\mathit{next}}
\newcommand{\send}{\mathit{send}}
\newcommand{\step}{\mathit{step}}
\newcommand{\rmsg}{\mathit{rmsg}}
\newcommand{\rmsgs}{\mathit{rmsgs}}
\newcommand{\smsgs}{\mathit{smsgs}}
\newcommand{\fdo}{\mathit{fdo}}
\newcommand{\fails}{\mathit{fails}}
\newcommand{\str}{\mathit{tr}}

\newcommand{\crs}{\mathit{CRS}}

\newcommand{\cralg}{\alg^R}
\newcommand{\crsys}{\mathcal{S}^R}
\newcommand{\crcfgs}{\mathcal{C}^R}
\newcommand{\crtr}{\mathit{Tr}}
\newcommand{\hd}{\mathit{head}}
\newcommand{\tl}{\mathit{tail}}
\newcommand{\last}{\mathit{last}}
\newcommand{\inv}{\mathit{Inv}}
\newcommand{\unfold}{\mathit{unfold}}
\newcommand{\frag}{\mathit{fr}}

\newcommand{\sper}{\mathit{stable}}

\newcommand{\pnet}{\mathcal{N}}
\newcommand{\pfail}{{F^\mathcal{P}}}
\newcommand{\epsnet}{\epsilon_{\mathcal{N}}}
\newcommand{\epsfail}{\epsilon_{F}}
\newcommand{\epsstable}{\epsilon_{\mathit{ss}}}

\newcommand{\buff}{\textit{buff}}
\newcommand{\acks}{\textit{acks}}

\newcommand{\inp}{\mathit{inp}}
\newcommand{\dec}{\mathit{dec}}

\usepackage{xcolor}
\usepackage{soulpos}
\makeatletter\AtEndDocument{\immediate\write\@auxout{\string\ulp@afterend}}\makeatother
\usepackage{pdfcomment}

\newcommand{\todo}[1]{\textit{\textcolor{red}{[TODO]: #1}}} 

\author{David Kozhaya}{ABB Corporate Research, Switzerland}{david.kozhaya@ch.abb.com}{}{}%mandatory, please use full name; only 1 author per \author macro; first two parameters are mandatory, other parameters can be empty.

\author{Ognjen Maric}{\text{ }}{ogi.yolmt@mynosefroze.com}{}{}

\author{Yvonne-Anne Pignolet}{ABB Corporate Research, Switzerland}{yvonne-anne.pignolet@ch.abb.com}{}{}

\authorrunning{D. Kozhaya, O. Maric, Y.A. Pignolet}%mandatory. First: Use abbreviated first/middle names. Second (only in severe cases): Use first author plus 'et al.'
\vspace{-8pt}
\Copyright{David Kozhaya, Ognjen Maric, Yvonne-Anne Pignolet}%mandatory, please use full first names. LIPIcs license is "CC-BY";  http://creativecommons.org/licenses/by/3.0/

\subjclass{}\vspace{-8pt}% mandatory: Please choose ACM 2012 classifications from https://www.acm.org/publications/class-2012 or https://dl.acm.org/ccs/ccs_flat.cfm . E.g., cite as "General and reference $\rightarrow$ General literature" or \ccsdesc[100]{General and reference~General literature}. 

%mandatory
\vspace{-8pt}
\EventEditors{}
\EventNoEds{0}
\EventLongTitle{}
\EventShortTitle{}
\EventAcronym{}
\EventYear{}
\EventDate{}
\EventLocation{}
\EventLogo{}
\SeriesVolume{}
\ArticleNo{}
\nolinenumbers %uncomment to disable line numbering
\hideLIPIcs  %uncomment to remove references to LIPIcs series (logo, DOI, ...), e.g. when preparing a pre-final version to be uploaded to arXiv or another public repository
\title{You Only Live Multiple Times:  A Blackbox Solution for Reusing Crash-Stop Algorithms In Realistic Crash-Recovery Settings}
\titlerunning{You Only Live Multiple Times}
\begin{document}
\maketitle

\vspace{-8pt}\begin{abstract}
Distributed agreement-based algorithms are often specified in a crash-stop 
asynchronous model augmented by Chandra and Toueg's unreliable failure 
detectors. In such models, correct nodes stay up forever, incorrect nodes 
eventually crash and remain down forever, and failure detectors behave correctly
forever eventually, However, in reality, nodes as well as communication links 
both crash and recover without deterministic guarantees to remain in some state 
forever.

In this paper, we capture this realistic temporary and probabilitic behaviour in
a simple new system model. Moreover, we identify a large algorithm class for 
which we devis a property-preserving transformation. Using this transformation, 
many algorithms written for the asynchronous crash-stop model run correctly and 
unchanged in real systems.

 \end{abstract}

\section{Introduction}
\label{sec:intro}
 
Distributed systems comprise multiple software and hardware components that are bound to eventually fail~\cite{Cristian:UFD}. Such failures can cause service malfunction or unavailability, incurring significant costs to
%, if not handled properly, can disrupt the system behavior and the availability of the services running on top. %While temporary service disruptions might be acceptable for certain users, other more
mission-critical systems, e.g., automation systems and on-line transactions. %, etc., incur significant costs relative to system failure or service unavailability events. 
%Machines in distributed systems fail and for many applications it is crucial that failures can be tolerated such that the system as a whole can continue to provide its service. 
The failures' impact can be minimized by protocols that let systems
agree on values and actions despite failures. As a consequence, many variants of the agreement or the
consensus problem~\cite{pease1980reaching} under different assumptions
have been studied. %originally proposed in
Of particular importance are synchrony and failure model assumptions,
as they determine the problem's complexity.

In the simplest failure model, often called the \emph{crash-stop}
model, a process fails by stopping to execute its protocol and never
recovers.  In an
asynchronous system, i.e., a system without bounds on execution
delays or message latency, assuming the \emph{crash-stop} failure model,
makes it impossible to
distinguish a crashed from a very slow process. This renders consensus-like problems unsolvable
deterministically~\cite{fischer1985impossibility}, already in this
very simple failure model. To circumvent this impossibility,
previous works have investigated ways to relax the underlying
asynchrony assumption either explicitly, e.g., by using partial
synchrony~\cite{DworkCPP}, or implicitly, by defining oracles that
encapsulate time, e.g., failure detectors~\cite{Chandra:UFD}. The
result is a large and rich body of literature that builds on top of
the former and latter techniques to solve consensus-like problems in
the presence of crash-stop~failures. Typically, the respective proofs
rely on assumptions of the ``eventually forever'' form: the correct
nodes staying up forever, incorrect nodes eventually crashing and
remaining down forever, and failure detectors producing wrong output
in the beginning, but eventually providing correct results forever.

However, such "eventually forever" assumptions are not met by real distributed systems. In reality,
processes may crash but their processors reboot, and the recovered
process rejoins the computation. Communication might also fail at any point in time, but get restored later. Hence, the failure and recovery modes of processes as well as
communication links are in reality probabilistic and
temporary~\cite{NSN,Fetzer,schmid2009impossibility},
especially in systems incorporating many unreliable
off-the-shelf low-cost devices and communication technologies. This led to the development of \emph{crash-recovery} models, where processes repeatedly leave and join the computation
unannounced. This requires new failure detector definitions and
new consensus algorithms built on top of these failure
detectors~\cite{aguilera2000failure,LARREA:CR,Dolev:1997,M.Hurfi:1998} as well 
as completely new solutions (without failure detectors) that consider different 
classes of failures, namely classified
according to how many times a process can crash and
recover~\cite{Mittal:2009}. However, such solutions eliminate the "eventually 
forever" assumptions only on the processes' level and not for the communication 
and failure detectors. Moreover, they do not tell us whether crash-stop 
algorithms can be ``ported unchanged'' to crash-recovery~settings.

To this end, this paper investigates how to re-use consensus algorithms defined for
the crash-stop model with reliable links and failure detectors in a
more realistic crash-recovery model, where processes and links can crash and recover probabilistically and for an unbounded number of times. Our models allow 
\emph{unstable} nodes, i.e., nodes that fail and recover infinitely
often. These are often excluded or limited in number in other models. In 
contrast, we explicitly allow unstable behavior of any number
of processes and links, by modeling
communication problems and crash-recovery behaviors as probabilistic
and temporary, rather than deterministic and perpetual. Our system model, 
similar to existing models that rely on probabilistic factors, e.g., coin 
flips, comes with the trade-off of solving consensus (namely the 
termination property), with probability 1, rather than deterministically. 

However, unlike existing solutions that incorporate probabilistic behavior, our approach does not aim at inventing new consensus algorithms but rather focuses on using existing deterministic ones to solve consensus with probability 1. Our approach is modular: we build a wrapper
that interacts with a crash-stop algorithm as a black box, exchanges
messages with other wrappers and transforms these messages into
messages that the crash-stop algorithm understands. We then formally
define classes of algorithms and safety properties for which we prove
that our wrapper constructs a system that preserves these
properties. Additionally, we show that termination with probability 1 is
guaranteed
for wrapped algorithms of this class. Moreover, this class is wide and includes
the celebrated Chandra-Toueg algorithm~\cite{Chandra:UFD} as well as the
instantiation of the indulgent framework with failure detectors
from~\cite{GUERRAOUI:2003} .
% (proven to be efficient and to have 0-degradation). 
Our work allows such algorithms to be ported unchanged
to our crash-recovery model. Hence applications built on top of such algorithms 
can run in real systems with crash-recovery behavior by simply using our wrapper.

\noindent\textbf{Contributions:}
To summarize, our main contributions are:
\begin{itemize}
	\item New system models that capture probabilistic and temporary failures and recoveries of processes and communication links in real distributed systems (described in Section~\ref{sec:system-models}) %(recovery, link failure) 
	\item A wrapper framework that allows a wide class of crash-stop consensus algorithms to be used unchanged in our more realistic models (described in
    Section~\ref{sec:wrappers-crash-stop})
    \item Formal properties describing which crash-stop consensus algorithms benefit from our framework and hence can be reused to solve consensus in crash-recovery settings (described Sections~\ref{sec:preservation-results} and~\ref{sec:ct-prob-term})
\end{itemize}

\noindent In addition to the sections presenting our contributions, we discuss related work in Section~\ref{sec: related work} and conclude the paper in Section~\ref{sec: conclusion}. %
\ifthenelse{\boolean{short}}{
	Due to space limitations, we
defer most proofs and some formalization details to a companion
technical report~\cite{TR}.}{}

%\noindent\textbf{Road-map:} The paper is organized as follows. Section~\ref{sec: related work} discusses related work. %In Section~\ref{sec:preliminaries} we define some terms, notations and concepts that we use to describe our system models and our solutions. 
%Section~\ref{sec:system-models} defines three systems types, (i)
%a crash-stop system with reliable links, (ii) a crash-recovery system
%with link failures, and (iii) a crash-recovery system with
%probabilistic crashes/recoveries of processes and links.
%Section~\ref{sec:wrappers-crash-stop} describes the wrapper 
%that transforms a crash-stop algorithm into an algorithm that
%tolerates crashes and recoveries. In
%Section~\ref{sec:preservation-results} we prove that our wrapper
%preserves a wide class of safety properties of crash-stop
%algorithms. In Section~\ref{sec:ct-prob-term} we prove that our
%wrapper also guarantees termination for a wide class of crash-stop
%consensus algorithms. %We formally define this class and show that it covers several notable consensus algorithms. 
%In Section~\ref{sec:
%  conclusion} we conclude the paper. 
%\ifthenelse{\boolean{short}}{
%	Due to space limitations, we
%defer most proofs and some formalization details to a companion
%technical report~\cite{TR}.}{}
%\yap{this paragraph could be shortened if necessary.}	

%%%%%%%%%%%%%%%%%%%%%%%%%%%%%%%%%%%%%%%%%%%%%%%%%%%%%%%%%%%%%%%
\section{Related Work}~\label{sec: related work} 
%%%%%%%%%%%%%%%%%%%%%%%%%%%%%%%%%%%%%%%%%%%%%%%%%%%%%%%%%%%%%%%
Several works
addressed the impossibility of asynchronous consensus. One direction
exploits the concept of partial synchrony~\cite{DworkCPP}, in which an
asynchronous system becomes synchronous after some unknown global
stabilization time (GST) for a bounded number of rounds. %In this case, consensus can be solved in this case if and only if $n \geq 2t + 1$. % Dwork et al.'s algorithm achieves consensus in $4(n+1)$ rounds.
%However, this perpetual synchrony is not necessary, e.g.,
For the same model, 
ASAP~\cite{alistarh2008solve} is a consensus algorithm where every
process decides no later than round $GST + f + 2$ (optimal). Another
direction augments asynchronous systems with failure detector oracles,
and builds asynchronous consensus algorithms on
top~\cite{Chandra:UFD}. These detectors typically behave erratically
at first, but eventually start behaving correctly forever.
Like with partial synchrony, the failure detectors must behave correctly for
only "sufficiently long" instead of forever~\cite{Chandra:UFD}; however, quantifying "sufficiently long" is not expressible in a purely asynchronous model~\cite{rethinkFD}.
%The folk belief that, like with partial synchrony, the failure detectors must behave correctly for
%only "sufficiently long" instead of forever~\cite{Chandra:UFD},
%a behavior not expressible in a purely asynchronous model~\cite{rethinkFD},
%has been shown correct in~\cite{NSN}.
Both lines of work initially investigated
crash-stop failures of processes. In real systems processes as well as
network links crash and recover multiple times and sometimes even
indefinitely. This gave rise to a large body of literature that
studied how to adapt the two lines of work to crash-recovery behavior
of processes and links.  We next survey some of this literature.

\textbf{Failure detectors and consensus algorithms for crash recovery:}
%[DFKM96, OGS97, HMR97, Aguilera1998] 
%
Dolev et al.~\cite{Dolev:1997} consider an asynchronous environment where communication links first lose messages arbitrarily, but eventually communication stabilizes such
that a majority  of processes forms a forever strongly connected component. Processes belonging to such a strongly-connected component are termed correct, and the others faulty. Process state is always (fully) persisted in stable storage. The authors propose a failure detector that allows the correct processes to reach a consensus decision and show that the rotating coordinator algorithm~\cite{Chandra:UFD} works unchanged in their setting, as long as all messages are constantly retransmitted. This relies on piggybacking all previous messages onto the last message, and regularly retransmitting the last message. As this yields very large messages, they also propose a modification of~\cite{Chandra:UFD} for which no piggybacking is necessary. While our results also rely on strongly connected components, we do not require their existence to be deterministic nor perpetual. We also do not require piggybacking in order for algorithms like~\cite{Chandra:UFD} to be used unchanged.
%where all processes (i) ``fast-forward'' the execution whenever they see a message from a round higher than the one they are executing and (2) only keep retransmitting the last message they sent to other processes.

Oliveira et al.~\cite{Oliveira:50073} consider a crash-recovery setting with correct processes
that may crash only finitely many times (and thus eventually stay up
forever) and faulty processes that permanently crash or crash
infinitely often. As in~\cite{Chandra:UFD}, the authors note that correct processes only need to stay up for long enough periods in practice (rather than forever), but this cannot be expressed in the asynchronous model. The authors take the consensus algorithm of~\cite{Schiper:1997} which uses stubborn links and transform it to work in the
crash-recovery setting by logging every step into stable storage and adding a fast-forward mechanism for skipping rounds. Hurfin et al.~\cite{Hurfin:995450} describe an algorithm using the $\evstrong$ detector in the crash-recovery case. The notions of correct/faulty processes and of failure detectors are the same as in Oliviera et al~\cite{Oliveira:50073}. Their algorithm is however more efficient when using stable storage compared to~\cite{Oliveira:50073}: there is only one write per round (of multiple data), and the required network buffer capacity for each channel (connecting a pair of processes) is one. Compared to~\cite{Oliveira:50073} and~\cite{Hurfin:995450} our system does not regard processes that crash and recover infinitely often as faulty and hence we allow such ``unstable'' processes to implement~consensus. 

Aguilera et al.~\cite{aguilera2000failure} consider a crash-recovery
system with lossy links. They show that previously proposed failure detectors 
for the crash-recovery setting have anomalous behaviors even in synchronous 
systems when considering unstable processes, i.e., processes that crash and 
recover infinitely often. The authors propose new failure detectors to mitigate 
this drawback. They also determine the necessary conditions regarding stable 
storage that allow consensus to be solved in the crash-recovery model, and 
provide two efficient consensus algorithms: one with, and one without using 
stable storage. Unlike~\cite{aguilera2000failure}, we do not exclude unstable 
processes from implementing consensus, thus our model tolerates a wider variety 
of node behavior. Furthermore, our wrapper requires no modifications to the 
existing crash-stop consensus algorithms, as it treats them as black-boxes.

\textbf{Modular Crash-Recovery Approaches:} 
Similar to~\cite{aguilera2000failure}, Freiling et al.~\cite{Freiling}
investigate the solvability of consensus in the crash-recovery model
under varying assumptions, regarding the number of unstable and correct
processes and what is persisted in stable storage. They reuse
existing algorithms from the crash-stop model in a modular way
(without changing them) or semi-modular way, with some modifications
to the algorithm (as in the case of~\cite{Chandra:UFD}). Similar
to our work, they provide algorithms to emulate a crash-stop system on top of a
crash-recovery system.  Our work, however, always reuses algorithms in a
fully modular way, and we define a wide class of algorithms for which such
reuse is possible. Furthermore, as we model message losses,
processes crashes, and process recoveries probabilistically, 
our results also apply if processes are unstable, i.e., crash and
recover infinitely often.

\textbf{Randomized Consensus Algorithms:} Besides the literature that studied deterministic consensus algorithms, existing works have also explored randomized algorithms to solve ``consensus with probability~1''. These include, for example, techniques based on using random \textit{coin-flips}~\cite{DAN:2014,Aspnes:2010,Rand3} or \textit{probabilistic schedulers}~\cite{Bracha:Toueg}. In systems with dynamic communication failures, multiple randomized algorithms~\cite{Randomizationhealer,Turquois} addressed the $k$-\textit{consensus} problem, which requires only $k$ processes to eventually decide. Moniz et al.~\cite{Randomizationhealer} considered a system with correct processes and a bound on the number of faulty transmission. In a wireless setting, where multiple processes share a communication channel, Moniz et al.~\cite{Turquois} devise an algorithm tolerating up to $f$ Byzantine processes and requires a bound on the number of omission faults affecting correct processes. In comparison, our work in this paper does not use randomization in the algorithm itself: we focus on using existing deterministic algorithms to solve consensus (with probability 1) in networks with probabilistic failure and recovery patterns of processes and links. %Moreover, we re-use existing deterministic algorithms relying on failure detectors to solve ``consensus with probability~1'' in~$\mathcal{N}$. 
%Easy consensus for the crash recovery model, Freiling, Lambertz, Majster, 2008, $https://ub-madoc.bib.uni-mannheim.de/1928/1/TR_2008_002.pdf$
%Technical report of 11
%Modular Consensus Algorithms for the Crash-Recovery Model, 2009 same authors as above
%$https://www1.informatik.uni-erlangen.de/filepool/publications/wras09-lambertz.pdf$

\section{System Models}
\label{sec:system-models}

We start with defining the notation we use, and then define general
concepts common to all of our models. Then, we define each of our
models in turn.

\noindent\textbf{Notation:}
Given a set $S$, we define $S_\bot$ to be the set $S \cup \set{\bot}$, where $\bot$ is a distinguished element not present in $S$.
% Given a pair $(a, b)$, we define the projections $\pi_1(a, b) = a$ and $\pi_2(a, b) = b$, and we
% similarly define $\pi_3$ for triples etc. 
The set of finite sequences over a set $S$ is denoted by
$S^*$. We also call
sequences \emph{words}, when customary. The
empty sequence is denoted by $[]$. Given a non-empty sequence, $\hd$
defines its first element, $\tl$ the remainder of
the sequence, and, if the sequence is finite, $\last$ its last element. Given two sequences $u$ and $v$, where
$u$ is finite, $u \cdot v$ denotes
their concatenation. For a word $u$, $\len{u}$ denotes the length of
$u$. Letting $u(i)$ be the $i$-th letter of $u$, we say that $u$ is a
\emph{subword} of $v$ if there exists a strictly monotone function
$f : \nats \rightarrow \nats$ such that $u(i) = v(f(i))$, for
$1 \le i \le \len{u}$ if $u$ is finite, and for all
$i \in \set{1, 2, \ldots}$ if $u$ is infinite. Analogously, $v$ is a \emph{superword} of $u$.

We denote the space of partial functions between sets $A$ and $B$ by
$A \pfun B$. Note that $(A \rightarrow B) \subseteq (A \pfun B)$. Given any function \mbox{$f : A \pfun B$}, its
\emph{graph}, written $\graph{f}$, is the relation
$\set{(x, y) \mid f(x) = y}$.
The \emph{range} of $f$, written $\ran(f)$, is the set $\set{y \mid (x, y) \in \graph{f}}$. 
% \emph{domain} of $f$, written
% $\dom(f)$ is the set $\set{x \mid (x, y) \in \graph{f}}$, and its

\noindent\textbf{Common concepts:}
We consider a fixed finite set of processes
$\procs = \set{1 \ldots N}$, and a fixed countable set of
values, denoted $\vals$.
For each algorithm there is an algorithm-specific countable set of
local states $\sspacep$, for each process $p \in \procs$. For
simplicity, we restrict ourselves to algorithms where $\sspacep = \sspaceq, ~\forall p,\, q \in \procs$.
Note that this does not exclude algorithms that take decisions based on 
identifiers. We define the global state space $\sspace = \prod_{p \in \procs} \sspacep$. Given a $s \in \sspace$, we
define $s_p \in \sspacep$ as the projection of $s$ to its $p$-th
component. %Finally, 
%for each algorithm $\alg$ 
%we assume  a countable message space $\msgs$.

%\subsection{Properties}
%\label{sec:properties}

A \emph{property} over an alphabet $A$ is a set of infinite words over
$A$.  We use standard definitions of liveness and safety
properties~\cite{alpern_defining_1985}.
	A property $P$ is a \emph{safety property} if, for
  every infinite word $w \notin P$, there exists a finite prefix $u$
  of $w$ such that the concatenation $u \cdot v \notin P$ for all infinite words $v$.
  Intuitively, the prefix $u$ is ``bad'' and not recoverable from.
	A property $P$ is a \emph{liveness property} if
  for any finite word $u$ there exists an infinite word $v$ such that
  $u \cdot v \in P$. Intuitively, ``good'' things can always happen
  later.

In this paper, we are interested in preserving properties over the
alphabet $\sspace$ between the crash-stop and crash-recovery versions
of an algorithm. In particular, we assume that the local states
$\sspacep$ are records, with two distinguished fields: $\inp$ of type
$\vals$ and $\dec$ of type $\vals_\bot$. Intuitively, a $\dec$
value of $\bot$ indicates that the process has not decided yet. For an
infinite word $w$ over the alphabet $\sspace$, let $w(i, p)$ denote
the local state of the process $p$ at the $i$-th letter of the word.
Let us state the standard safety properties of consensus in our notation.

\begin{description}
\item[Validity.] Decided values must come from the set of input
  values. Formally, validity describes the set of
  words $w$ such that
  $
    \forall p,\ i,\ v.\ w(i, p).\dec = v \land v \neq \bot \Longrightarrow 
    \exists q.\ w(1, q).\inp = v
  $
\item[Integrity.] Processes do not change their decisions. Formally, integrity describes the set of words $w$ such that
  $
    \forall p,\ i.\ v.\ w(i, p).\dec = v \land v \neq \bot \Longrightarrow 
     \forall i' > i.\ w(i', p).\dec = v
  $
\item[(Uniform) Agreement.] No two processes ever make different non-$\bot$
  decisions. Formally,
  $
    \forall p,\ q,\ i,\ j.\ w(i, p).\dec \neq \bot \neq w(j, q).\dec \Longrightarrow 
    w(i, p).\dec = w(j, q).\dec
  $

\end{description}

To simplify our preservation results for safety properties,
our models store information about process failures separately
from $\sspace$. %Namely, if $\sspace$ included failure information, then the following simple safety property would not be preserved between crash-stop and crash-recovery settings: once a process fails, it never goes back to an ``unfailed'' state.
As a consequence, the standard crash-stop termination property cannot be
expressed as a property over $\sspace$: it is conditioned on a
process not failing. 
%We would have to either enlarge $\sspace$ with the currently failed processes, or express termination over a
%different alphabet. 
However, we do not directly use the crash-stop notion
of termination and we omit this definition here. %for the moment
%, as we expect that, for crash-recovery, we will be able to 
Instead, we will prove the following property for the algorithms in our probabilistic crash-recovery model:

\begin{description}
\item[Probabilistic crash-recovery termination.] With probability $1$,
  all processes eventually decide. 
\end{description}

\subsection{The crash-stop model}
\label{sec:crash-stop-model}

Our definition of the crash-stop model is standard and closely
follows~\cite{Chandra:UFD}. We assume an asynchronous environment,
with processes taking steps in an interleaved fashion. Processes
communicate using reliable links, and can query failure detectors.

\noindent\textbf{Failure detectors:}
%\label{sec:failure-detectors}
A \emph{failure pattern} $\mathit{fp}$ is an infinite word over the alphabet
$2^\procs$. Intuitively, 
each letter is the set of failed processes in a 
transition step of a run of a transition
system. A \emph{failure detector} with range $\fdran$ is a function from failure
patterns to properties over the alphabet $\fdran$.%
\footnote{This definition does not distinguish which process received the output, which is sufficient for $\evstrong$. The definition can be easily extended to other failure detectors like $\evweak$.} 
%\ogi{This definition does not distinguish which process received the output. This should be sufficient for $\evstrong$, though not for $\evweak$. It should be easy enough to change later.} 
A failure detector $D$ is 
\emph{unreliable} if $D(\mathit{fp})$ is a liveness property for all $\mathit{fp}$. 
Intuitively, a detector constrains how
the failure detector outputs (the $\fdran$ values) must depend on the
failure pattern of a run, and
unreliable detectors can produce arbitrary outputs in the beginning. 
 We write $\fds(\fdran)$
for the set of all detectors with range $\fdran$.
%\yap{is this really the definition we wanna use?}

\noindent\textbf{Algorithms and algorithm steps:}
\label{sec:algorithm-steps}
The type of \emph{crash-stop steps} over a message space $\msgs$ and
a failure detector range $\fdran$, written $\css(\msgs, \fdran)$ is defined
as a pair of functions of types:

\ifthenelse{\boolean{short}}{
$\nnext: \sspacep \times (\procs \times \msgs)_\bot \times \fdran \rightarrow
  \sspacep $, ~~ $\send:  \sspacep \rightarrow (\procs \pfun \msgs). $
}
{
\begin{itemize}
  \item $\nnext: \sspacep \times (\procs \times \msgs)_\bot \times \fdran \rightarrow
  \sspacep $
  \item $\send:  \sspacep \rightarrow (\procs \pfun \msgs). $
\end{itemize}
}

\noindent
Intuitively, given zero or one messages received from some other
processes and an output of the failure detector, a step maps the current
process state to a new state, and maps the new state to a set
of messages to be sent, with zero or one messages sent to each
process.

\noindent A \emph{crash-stop
  algorithm} $\alg$ over $\sspace$, $\msgs$ and $\fdran$ is a tuple
$(I, \step, D, N_f)$ where:
\ifthenelse{\boolean{short}}{
$I \subseteq \sspace$ is the \emph{finite} set of initial states, $\step \in \css(\msgs, \fdran)$ is the step function,
  and 
$D \in FD(\fdran)$ is a failure detector.
 $N_f < N$ is the resilience condition, i.e., the number of
  failures tolerated by the algorithm (recall that we consider a fixed
  $N$).
}
{
\begin{itemize}
\item $I \subseteq \sspace$ is the \emph{finite} set of initial states,
\item $\step \in \css(\msgs, \fdran)$ is the step function,
  and 
\item $D \in FD(\fdran)$ is a failure detector.
\item $N_f < N$ is the resilience condition, i.e., the number of
  failures tolerated by the algorithm (recall that we consider a fixed
  $N$).
\end{itemize}
}
We refer to the components of an algorithm $\alg$ by $\alg.I$, $\alg.\step$, $\alg.D$ and $\alg.N_f$.
\ifthenelse{\boolean{short}}{}
{
As $I$ is finite, it will admit a uniform distribution in our
probabilistic model in Section~\ref{sec:probabilistic-cr}.
}

\noindent\textbf{Configurations:}
\label{sec:configurations}
As noted earlier, we focus on preserving properties over $\sspace$
between crash-stop and crash-recovery models.
However, $\sspace$ contains insufficient information to model the
algorithm's crash-stop executions (runs). In particular, to account for
\begin{enuminline}
\item asynchronous message delivery and
\item process failures,
\end{enuminline}
we must extend states to \emph{configurations}. A \emph{crash-stop configuration}
is a triple $(s, M, F)$ where:
	\ifthenelse{\boolean{short}}{
	$s \in \sspace$ is the (global) state,
$M \subseteq \procs \times \procs \times \msgs$ is the set of
  \emph{in-flight messages}, where $(p, q, m) \in M$ represents a
  message $m$ that was sent to $p$ by $q$, and  $F \subseteq \procs$ is the set 
	of failed processes.
	}{
\begin{itemize}
\item $s \in \sspace$ is the (global) state,
\item $M \subseteq \procs \times \procs \times \msgs$ is the set of
  \emph{in-flight messages}, where $(p, q, m) \in M$ represents a
  message $m$ that was sent to $p$ by $q$.
For simplicity, $M$ is
  a set, i.e., we assume that each message is sent at most once
  between any pair of processes during the entire execution of the
  algorithm. This suffices for round-based algorithms that tag messages
  with their round numbers, and exchange messages once per round.
	%Note that this implies that no process ever reaches the same
  %local state $s_p$ with a non-empty $\send(s_p)$ twice.\ogi{Is this remark used later? Delete?}
\item $F \subseteq \procs$ is the set of failed processes.
\end{itemize}
}
\noindent As with algorithms, we refer to the components of a configuration $c$ by
$c.s$, $c.M$ and $c.F$.

\noindent\textbf{Step labels and transitions:}
\label{sec:transitions}
While the algorithm steps are deterministic, the asynchronous transition system is not: any
(non-failed) process can take a step at any point in time, with
different possible received messages, and different failure detector
outputs. Accessing this non-determinism information is useful in proofs, so we extract it as follows. %\ogi{Check if we do need it.}
\ifthenelse{\boolean{short}}{}
{\\}
A \emph{crash-stop step label} is a quadruple $(p, \rmsg, \fails, \fdo)$,
where:
\ifthenelse{\boolean{short}}{
$p \in \procs$ is the process taking the step,
$\rmsg \in (\procs \times \msgs)_\bot$ is the message $p$ receives in the 
step ($\bot$ modeling a missing message),
$\fails \subseteq \procs$ is the set of processes failed
  at the end of the step,
$\fdo \in \fdran$ is $p$'s output of the failure detector.
}
{\begin{itemize}
\item $p \in \procs$ is the process taking the step,
\item $\rmsg \in (\procs \times \msgs)_\bot$ is the message $p$ receives in the 
step ($\bot$ modeling a missing message),
\item $\fails \subseteq \procs$ is the set of processes failed
  at the end of the step,
\item $\fdo \in \fdran$ is $p$'s output of the failure detector.
\end{itemize}
}
A \emph{crash-stop step} of the
algorithm $\alg$ is a triple $(c, l, c')$, where $c$ and $c'$ are configurations and $l$ % = (p, \rmsg, \fails, \fdo)$
 is a label. Crash-stop steps must satisfy the following properties: 
\ifthenelse{\boolean{short}}{
(i) $p$ is not failed at the start of
  the step. 
(ii) $p$ takes a step according to the label and the algorithm's rules, and the 
other processes do not move. 
(iii)
 If a message was received, then it was in flight; the received message is
removed from the set of in-flight messages, while the produced messages are 
added.
(iv) Failed processes do not recover, i.e., $c.F \subseteq c'.F = \fails$.
}
{\begin{itemize}
\item $p \notin c.F$, i.e., $p$ is not failed at the start of
  the step.

\item With $s_p$ and $s_p'$ denoting process $p$'s state in $c$ and
  $c'$ respectively, $s_p' = \nnext(s_p, \rmsg, \fdo)$, and
  $s_q' = s_q$ for $q \neq p$. That is, $p$ takes a step according to
  the label and the algorithm's rules, and the other processes do not
  move. 
	%% First part a candidate for long?

\item If $\rmsg \neq \bot$, then $(p, \rmsg) \in c.M$, i.e., if
  a message was received, then it was in flight.

\item Letting $\smsgs = \send(s_p')$, then
  $
    c'.M = (c.M\, \setminus\, \mathit{rcvd}) \cup \set{(q, p, m) \mid (q,
      m) \in \graph{\smsgs}},
  $
  where $\mathit{rcvd}$ is
  $\set{(p, \rmsg)}$ if $\rmsg \neq \bot$ and $\emptyset$ otherwise.
  That is, the received message is removed from the set of in-flight
  messages, and the produced messages are added.

\item $c.F \subseteq c'.F$. That is, failed processes do not recover
  in the crash-stop setting.

\item $\fails = c'.F$.

\end{itemize}
}

\noindent\textbf{Algorithm runs:}
\label{sec:algorithms-runs}
A finite, respectively infinite \emph{crash-stop run} of $\alg$ is a
finite, respectively infinite alternating sequence
$c_0,\, l_0,\, c_1 \ldots$ of configurations and labels, that ends
in a configuration if finite, such~that
\ifthenelse{\boolean{short}}{
the initial state is allowed by the
  algorithm, $(c_i, l_i, c_{i+1})$ are valid steps, the resilience
  condition is satisfied. Furthermore,  the output of the failure detector must 
	satisfy the condition of the failure detector.
}{
\begin{itemize}
\item $c_0.s \in \alg.I$, i.e., the initial state is allowed by the
  algorithm.
\item For each $c_i$, $l_i$ and $c_{i+1}$ in the sequence,
  $(c_i, l_i, c_{i+1})$ is a step of $\alg$.  
  
\item For each $c_i$ in the sequence, $\card{c_i.F} \le N_f$ (the resilience
  condition is satisfied).

\item There is a failure pattern $\mathit{fp}$ such that
  the sequence $c_i.F$ is a prefix of $\mathit{fp}$, and
  the sequence $\fdo_i$ is a prefix of $D(fp)$.
  That is,
  the output of the failure detector satisfies the
  condition of the failure detector.
\end{itemize}
}
Such a run has \emph{reliable links} if,
\ifthenelse{\boolean{short}}{
all in-flight messages eventually get delivered, unless the
sender or the receiver is faulty.
}{
whenever $(p, q, m) \in c_i.M$, and $p, q \notin c_j.F$ for
all $j$, then there exists a $k \ge i$ such that
$l_k = (p, (q, m), \fails, fdo)$ for some $\fails$ and $\fdo$. That
is, all in-flight messages eventually get delivered, unless the
sender or the receiver is faulty.
}
The \emph{crash-stop system} of the algorithm $\alg$ is the sequence
of all crash-stop runs of $\alg$. The \emph{crash-stop system with
  reliable links} of the algorithm $\alg$ is the set of all crash stop
runs with reliable links. 

As mentioned before, we are interested in properties that are sequences of global states. In this sense, runs
contain too much 
information (e.g., in-flight messages). Thus, given a run $\rho$, we define its \emph{state
  trace} $\str(\rho)$, obtained by removing the labels and projecting
configurations onto just the states. We introduce a notion of a \emph{state property}: an
infinite sequence of (global) states. The crash-stop system (with or without
reliable links) satisfies a state property $P$ if for every run $\rho$
of the system, $\str(\rho) \in P$. We later show that our crash-recovery
wrappers for crash-stop protocols preserve important state properties of
crash-stop algorithms. Lastly, we note down some simple properties of crash-stop runs.

\begin{lemma}[Reliable links irrelevant for prefixes]
  \label{lemma:reliable-irrelevant-prefix}
  Let $\rho$ be a finite crash-stop run of $\alg$. Then, $\rho$ can be
  extended to an infinite crash-stop run of $\alg$ that has reliable
  links, by extending $\str(\rho)$ to include infinite message retransmissions for all sent messages. 
\end{lemma}
\vspace{-20pt}\paragraph*{Summary of time and failure assumptions}
\vspace{-5pt}\textbf{Time}.
Processes are asynchronous and have no notion of time. Links are asynchronous.

\noindent\textbf{Failures}.
Processes can fail by halting forever, while links interconnecting them do not fail.

%\begin{lemma}[Reliable links irrelevant for safety]
 %For any algorithm $\alg$, its crash-stop system with reliable links satisfies a safety property if and only if its crash-stop system does.
%\end{lemma}

% \begin{proof}
%   One direction is trivial; the other follows from
%   Lemma~\ref{lemma:reliable-irrelevant-prefix} and the definition of
%   safety properties.
% \end{proof}

\subsection{The lossy synchronous crash-recovery model}
\label{sec:non-deterministic-cr}

We next define our first crash-recovery model. Formally, this a lossy
synchronous crash-recovery model, with non-deterministic, but not
probabilistic losses, crashes, and recoveries. We use it to prove
the preservation of safety properties without taking probabilities
into account, since they are not used in such arguments.
In this model, we will not distinguish between volatile and persistent
memory of a process. Instead, we assume that all memory is persistent.
This can be emulated in practice by persisting all volatile memory
before taking any actions with side-effects (such as sending network
messages). 
%\yap{we should not forget to investigate this later}
Finally, while the model is formally synchronous, in that all processes take
steps simultaneously, it also captures processing delays, as a slow
process behaves like a process that crashes and later recovers.

\noindent\textbf{Algorithms and algorithm steps:}
\label{sec:algor-algor-steps}
A \emph{crash-recovery step} over a message space $\msgs$,
written $\crs(\msgs)$ is defined as: %as the pair of functions:

\ifthenelse{\boolean{short}}{
 $\nnext: \sspacep \times (\procs \pfun \msgs) \rightarrow \sspacep$
~~ $\send: \sspacep \rightarrow (\procs \rightarrow \msgs).$
}{
\begin{itemize}
  \item $\nnext: \sspacep \times (\procs \pfun \msgs) \rightarrow \sspacep$
  \item $\send: \sspacep \rightarrow (\procs \rightarrow \msgs).$
\end{itemize}
}

\noindent
In other words, a step determines the new state based on the current state and the set of received messages. Given the new state, a process sends a message to every other processes (including itself). Compared to
the asynchronous setting (Section~\ref{sec:crash-stop-model}), in this model:
\begin{enumerate}
\item  A process can receive multiple messages simultaneously (rather than receiving at most one message in a step).
  
\item Every process sends a message to every other process at each
  step. %This slightly simplifies the definition of the probabilistic model in Section~\ref{sec:probabilistic-cr}, and will hold for the algorithms we focus on, namely ``wrapped'' versions of crash-stop algorithms.
  %Intuitively, at every step, each pair of processes will exchange
 We use this in later sections to send heartbeat messages, if there is nothing else to exchange. We do not require any guarantees on the delivery of the sent messages in this model. 
 
  %These guarantees are probabilistic, and are modeled by a Markov chain in Section~\ref{sec:probabilistic-cr}.

\item No failure detector oracle is specified. The synchrony assumption of this model inherently provides spurious failure detection: each process can suspect all peers it
  did not hear from in the last message exchange. This is in fact
  exactly what we will use to provide failure detector outputs to the
  ``wrapped'' crash-stop algorithms run in this setting\footnote{Similar to Gafni's round-by-round fault
  detectors~\cite{gafni_round-by-round_1998}; in our case, the detectors
  are ``step-by-step''}.
\end{enumerate}

\noindent A \emph{crash-recovery algorithm} $\cralg$ over $\sspace$, $\msgs$ is
a pair $(I, \step)$ where:
\ifthenelse{\boolean{short}}{
$I \subseteq \sspace$ is a finite set of initial states, and
$\step \in \crs(\msgs)$ is the step function.
}{
\begin{itemize}
\item $I \subseteq \sspace$ is a finite set of initial states, and
\item $\step \in \crs(\msgs)$ is the step function.
\end{itemize}
}

\noindent\textbf{Configurations:}
\label{sec:configurations-cr}
As in the crash-stop case, we require more than just the global states
to model algorithm executions; we hence introduce configurations. These, however,
differ from those for the crash-stop setting. As
communication is synchronous, we need not store the in-flight
messages; they are either delivered by the end of a step, or they are
gone. Furthermore, as processes take steps synchronously, we
can introduce a global step number.~\footnote{We use global step numbers later in
the probabilistic model, to assign failure probabilities for processes and
links (e.g., a probability $p_{ij}(t)$ of a message from $i$
to $j$ getting through, if sent in the $t$-th step).} 

A \emph{crash-recovery configuration}
is a tuple $(n, s, F)$ where
\ifthenelse{\boolean{short}}{
$n \in \nats$ is the step number,
$s \in \sspace$ is the (global) state,
$F \subseteq \procs$ is the set of failed processes.
}{\begin{itemize}
\item $n \in \nats$ is the step number,
\item $s \in \sspace$ is the (global) state,
\item $F \subseteq \procs$ is the set of failed processes.
\end{itemize}
}
We denote the set of all crash-recovery configurations by $\crcfgs$.
Note that this set is countable. This will allow us to impose a Markov chain 
structure on the system in the later model.

\noindent\textbf{Step labels and transitions:}
\label{sec:transitions-cr}
As in the crash-stop setting, we use labels to capture all sources
of non-determinism in a step. We will use
these labels to assign probabilities to different state transitions in
the probabilistic model of the next section.

A \emph{crash-recovery step label} is a pair $(\rmsgs, \fails)$,
where:
\begin{itemize}
\item $\rmsgs : \procs \rightarrow (\procs \pfun \msgs)$ denotes the message 
	received in the step; $\rmsgs(p)(q)$ is the message received by $p$ on the
  channel from $q$ to $p$. Note that the function is partial. As we assume that
	$q$ always attempts to send a message to $p$, if $\rmsgs(p)(q)$ is undefined 
	($\rmsgs(p)$ is a partial function), then either the message on this channel 
	was lost in the step, or the sender $q$ has failed.

\item $\fails \subseteq \procs$ is the set of processes that are failed
  at the end of the step.
\end{itemize}

The \emph{crash-recovery steps} (or transitions) of $\alg$, written
$\crtr(\alg)$, is the set of all triples $(c, l, c')$,
\ifthenelse{\boolean{short}}{
where
   (i) only processes that are up handle their messages (ii) messages
  from failed senders are not received (iii) failed processes $\fails=c'.F$.
}{
 where,
letting $l = (\rmsgs, \fails)$ and letting $s_p$ and $s_p'$ denote the state of 
the process $p$ in $c$ and $c'$, we have:
\begin{itemize}
    
\item if $p \notin c.F$, then $s_p' = \nnext(s_p, \rmsgs(p))$; that
  is, processes that are up handle their messages. We assume that
  the state change is atomic; this can be implemented, since we
  assume that all process memory is persistent.
    
\item if $p \in c.F$, then $s_p' = s_p$. Failed processes do not
  change their state. 

\item if $\rmsgs(p)(q)$ is defined, then $q \notin c.F$,
  and $\rmsgs(p)(q) = \send(s_q)$. That is, only the sent step
  messages are received (no message corruption), and messages
  from failed senders are not received.
   
\item $\fails = c'.F$.

\end{itemize}
}

\noindent\textbf{Algorithm runs:}
\label{sec:algorithm-runs-cr}
A finite, resp. infinite \emph{crash-recovery run} of $\alg$ is a
finite, resp. infinite alternating sequence
$c_0,\, l_0,\, c_1 \ldots$ of (crash-recovery) configurations and labels,
ending in a configuration if finite, such that:
\ifthenelse{\boolean{short}}{
$c_0.s \in \alg.I$, i.e., the initial state is allowed by the
  algorithm,
and $(c_i, l_i, c_{i+1}) \in \crtr(\alg)$ for all $i$, that
  is, each step is a valid crash-recovery transition of $\alg$.
}{
\begin{itemize}
\item $c_0.s \in \alg.I$, i.e., the initial state is allowed by the
  algorithm
\item $(c_i, l_i, c_{i+1}) \in \crtr(\alg)$ for all $i$, that
  is, each step is a valid crash-recovery transition of $\alg$.
\end{itemize}
}
The \emph{crash-recovery system} of algorithm $\alg$ is the
sequence of all crash-recovery runs of~$\alg$. %Analogous to
% the crash-stop case, we define a \emph{crash-recovery configuration
%   property} as a set of alternating sequences of crash-recovery
% configurations and labels. %However, we defer the definition of
%state traces. To run crash-stop algorithms in the crash-recovery
%model, we need a wrapper. The wrapper will use some additional
%state, not present in the crash-stop runs. To preserve crash-stop
%properties, we will disregard this additional state.

\vspace{-5pt}\paragraph*{Summary of time and failure assumptions}
\textbf{Time}.
Processes are synchronous and operate in a time-triggered fashion. Links are synchronous (all delivered messages respect a timing upper bound on delivery).

\noindent\textbf{Failure}.
Processes can fail and recover infinitely often. In every time-step, a link can be either crashed or correct. A crashed link drops all sent messages (if any).
\subsection{The probabilistic crash-recovery model}
\label{sec:probabilistic-cr}
We now extend the lossy synchronous crash-recovery model to a probabilistic 
model, where both the successful delivery of messages and failures follow a 
distribution that can vary with time.
A \emph{probabilistic network} $\pnet$ is a function of type
$\procs \times \procs \times \nats \rightarrow [0, 1]$, such that 
$
  \exists \epsnet > 0.\ \forall p,\, q,\, t.\ \pnet(q, p, t) > \epsnet
$
Intuitively, $\pnet(q, p, t)$ is the delivery probability for a
message sent from $p$ to $q$ at time (step number) $t$. A \emph{probabilistic 
failure pattern} $\pfail$ is a function
$\procs \times \nats \rightarrow [0, 1]$, such that
$
\exists \epsfail > 0.\ \forall p,\, t.\ \epsfail < \pfail(p, t) < 1 - \epsfail. 
$
Intuitively, $\pfail(p, t)$
gives the probability of $p$ being up at time $t$.\footnote{Considering infinite time, the upper and lower bounds on $\pfail(p, t)$ ensure that, with probability~1, there is a time when process $p$ is up.}

Given a crash-recovery algorithm $\cralg$, a probabilistic network
$\pnet$ and failure pattern $\pfail$, a
\emph{probabilistic crash-recovery system}
$\crsys(\cralg, \pnet, \pfail)$ is the Markov chain~\cite{baier_principles_2008} with:
\begin{itemize}
\item The set of states $\crcfgs$, i.e., the crash-recovery configuration set.
	
\item The transition probabilities $P(c)(c')$ defined as
  $P(c) = \mathit{nf_\mathit{trans}(c)} \cdot \mathit{trans(c)}$, where $\mathit{nf_\mathit{trans}(c)}$ is the normalization factor for $\mathit{trans}(c)$ and $\mathit{trans}(c)(c')$ is defined as:
  \begin{align*}
    \hspace*{-30pt}\mathit{trans}(c)(c') = & \sum\limits_{\rmsgs, \fails} \truthof{(c, (\rmsgs, \fails), c') \in \crtr(\cralg)}  
                 \cdot \prod\limits_{p, q} (\truthof{\rmsgs(p)(q) \text{ defined}}
                 \cdot \pnet(q, p, c.n) \\
               &\qquad \hspace*{-30pt}+ \truthof{\rmsgs(p)(q) \text{ undefined}}\cdot\truthof{ q \notin c.F} \cdot (1 - \pnet(q, p, c.n))
                 )
                  \\
               & \hspace*{-9pt}\cdot  
                 \prod\limits_{p} ( (1 - \pfail(p,\, c'.n)) \cdot \truthof{p \in \fails} + \pfail(p,\, c'.n) \cdot \truthof{p \notin \fails}
                 ).
  \end{align*}
	Here, $\truthof{\cdot}$ maps the Boolean values true and false to $1$ and $0$ respectively.
  Intuitively, a transition from $c$ to $c'$ is only possible if it is
  possible in the lossy synchronous crash-recovery model. The
  probability of this transition is calculated by summing over all
  labels that lead from $c$ to $c'$, and giving each such label a
  weight. Note that the only non-determinism in the transitions of
  $\crtr(\cralg)$ comes exactly from the behavior we deem probabilistic:
  messages being dropped by the network, and process failures.

  It is easy to see that $\mathit{nf}_\mathit{trans}(c)$ is well
  defined for all $c$, as for a fixed configuration $c$,
  $\mathit{trans}(c)(c')$ is non-zero for only finitely many
  configurations $c'$. 
	
	%We require normalization since the sum:
  %\begin{align*}
    %\sum\limits_\rmsgs & \truthof{\exists \fails, c'.\ (c, (\rmsgs. fails), c') \in \crtr(\cralg)} \cdot \\
    %& \prod\limits_{p, q} (\truthof{\rmsgs(p)(q) \text{ defined}}
    %\cdot \pnet(q, p, c.n) \\
    %& \qquad + \truthof{\rmsgs(p)(q) \text{ undefined}}\cdot\truthof{ q \notin c.F}
    %\cdot (1 - \pnet(q, p, c.n)) )
  %\end{align*}
  %can be smaller than $1$. In fact, this sum is the normalization
  %factor $\mathit{nf_\mathit{trans}}$. 
  %
\item The distribution over the initial states defined by $\mcinit = \norm(\mathit{init})$, where
  \begin{itemize}
  \item $P_f(F) = \prod\limits_{p \notin F} \pfail(p, 0) \cdot \prod\limits_{p \in F} (1 - \pfail(p, 0))$,
  \item $
    \mathit{init}(n, s, F) = P_f(F) \cdot \truthof{n = 0}\cdot\truthof{s \in \cralg.I},
  $
  \end{itemize}
  and $\norm$ normalizes the probabilities. Note that normalization is possible,
  since we assumed that after fixing $N$, each algorithm comes with a finite set of initial
  states.
  %\ogi{Other possibilities:
   % \begin{enumerate}
    %\item assume that some probability distribution on the initial  states is given, and that it matches $\pfail$ (or redefine $\pfail$ such that it only works for time points after the initial one).

%\item do some kind of weird Markov decision process, with only one non-deterministic choice, of the initial state. That is, we would look at the probability of each infinite trace, given a fixed initial state.
 %\end{enumerate}
  %  }
\end{itemize}
\vspace{-5pt}\paragraph*{Summary of time and failure assumptions}
\vspace{-5pt}\textbf{Time}.
Processes are synchronous and operate in a time-triggered fashion. Links are synchronous (all delivered messages respect a timing upper bound on delivery).

\noindent\textbf{Failure}.
Processes and links can fail and recover infinitely often. At the beginning of any time-step a crashed process/link can recover with positive probability and a correct process/link can fail with positive probability.

\vspace{5pt}\noindent\textbf{Important:} All our results (Section~4,~5, and 6) hold for 
both crash-recovery models (Section~\ref{sec:non-deterministic-cr}
 and~\ref{sec:probabilistic-cr} ). However, probabilities are only needed to 
prove results of Section~\ref{sec:ct-prob-term}.

\section{Wrapper for Crash-Stop Algorithms}
\label{sec:wrappers-crash-stop}

We now define the transformation of a crash-stop
algorithm $\alg$ into a crash-recovery algorithm $\alg^R$. Intuitively,
we do this by (also illustrated in Figure~\ref{fig:wrapper}):
\begin{wrapfigure}{R}{5.2cm}
	\centering
	\vspace{-.2cm}
		\includegraphics[width=5.2cm]{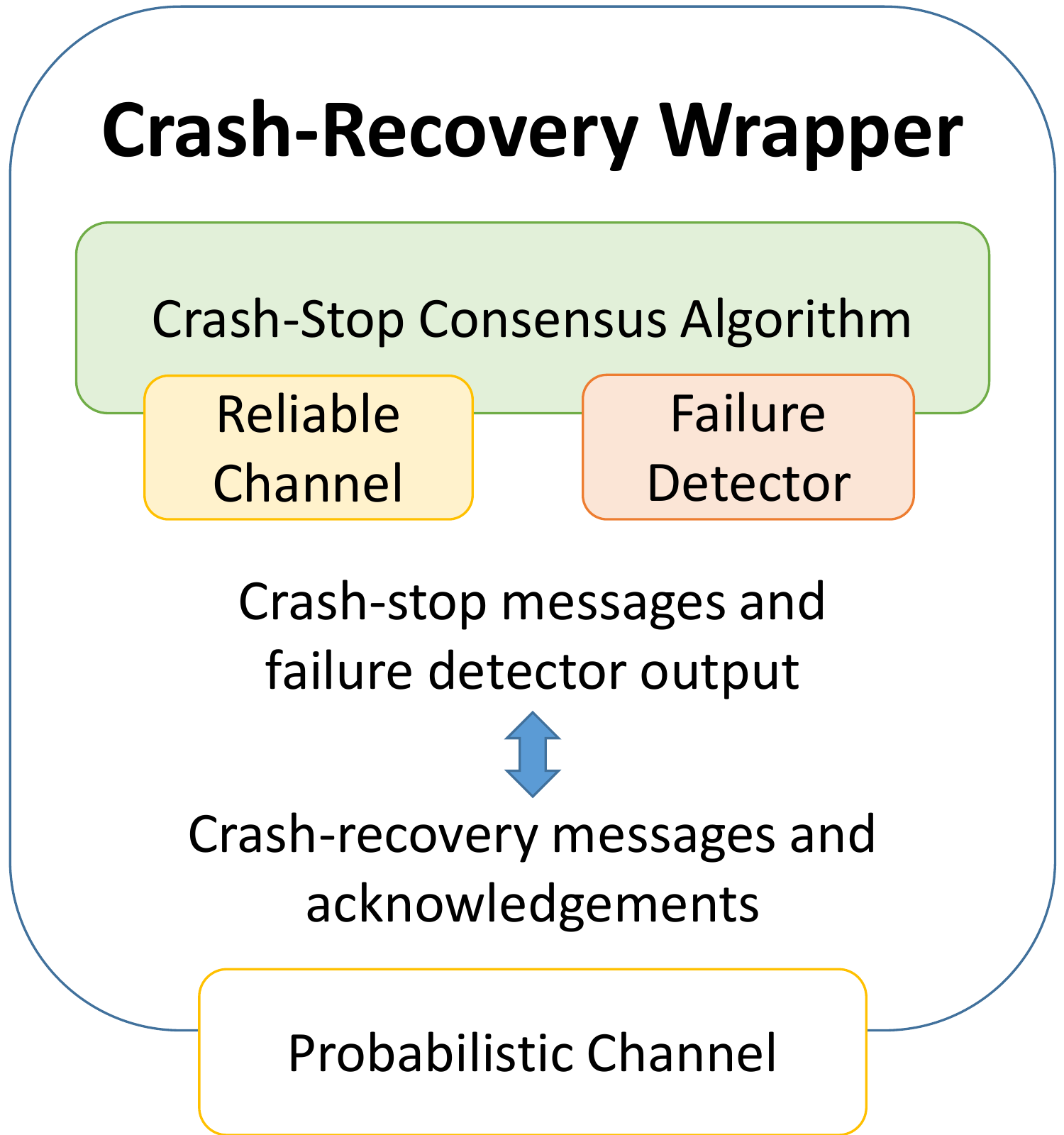}
	\caption{Wrapper concept.}
	\vspace{-.2cm}
	\label{fig:wrapper}
\end{wrapfigure}
\begin{itemize}[noitemsep,topsep=0pt]
\item Generating a synchronous crash-recovery step using a series of
  crash-stop steps. % Namely, in the crash-recovery model, all processes
%  take steps simultaneously, with each process handling multiple
 % incoming messages. 
  Each step in the series handles one individual received message,
  allowing us to iteratively handle multiple simultaneously incoming messages
  and bridge the synchrony mismatch between the crash-stop and crash-recovery
  models.
  
\item Using round-by-round failure detectors to produce the failure detector
  outputs to be fed to the crash-stop algorithm. These outputs are from the set 
	$2^\procs$.\footnote{We could instead produce outputs that never suspect 
	anyone, since no process crashes forever in our probabilistic model. 
	However for a weaker model
	\ifthenelse{\boolean{short}}{ 
	that we define in the technical report~\cite{TR}
	}{
	considered later
	}
	(where 
	processes are allowed to crash forever), we need failure detectors that 
	suspect processes.}
  % Note that for arbitrary finite sequences of unreliable failure detector outputs, there is a legal run of the algorithm matching this output.\ogi{The previous sentence seems disconnected}
  
\item Providing reliable links, as required by the crash-stop algorithm. During each crash-recovery step,
  we buffer all outgoing messages of a process,
  and send them repeatedly in the subsequent crash-recovery steps, until an
  acknowledgment is received.
  % We use LIFO buffers, which is crucial for our preservation result later.
\end{itemize}

\noindent
We first define the message and state spaces of the crash-recovery
version~$\alg^R$ of a given crash-stop algorithm~$\alg$ as follows:

\begin{itemize}
\item In~$\alg^R$, we send a pair of messages to each 
	process in each step: the actual \emph{payload} message (from $\alg$), replaced by a special \emph{heartbeat}
    message $\hbeat$ being sent when no payload needs to be sent; and an \emph{acknowledgment} message, confirming the receipt of
    the last message on the channel in the opposite direction, 

  \item The local state $s_p$ of a process $p$ has three components $(st, \buff,\acks)$:
    (1) $st$ stores the state of $p$ in the  target crash-stop algorithm; (2) $\buff$ represents $p$'s outgoing message buffers, with one buffer for each  process (including one for $p$); and (3) $\acks(q)$ records
    the last message that $p$ received from $q$. The buffers are LIFO, a choice which proves crucial
    for our termination proof (Section~\ref{sec:ct-prob-term}).
\end{itemize}

\noindent Next, given a crash-recovery state $s$, a process $p$, and the messages $\rmsgs$ received by $p$ in the given round, we define $\unfold(p, s, \rmsgs)$,
$p$'s local step unfolding for $s$ and $\rmsgs$.
We define $\unfold(p, s, \rmsgs)$ as the sequence of intermediate steps $p$ takes.  Said differently, $\unfold(p, s, \rmsgs)$ is a sequence of 
crash-recovery states and crash-stop labels
$s_0, l_0, s_1, l_1, s_2, \ldots s_n$, where the intermediate state $s_i$
represents the state of $p$ after processing the message 
of the $i$-th process. In a crash-recovery run, $p$ does not actually transition through the  states $s_1, \ldots, s_n$. These states are listed here separately to intuitively show how the next ``real'' state to which $p$ will transition is computed. The unfolding also allows to relate traces of $\alg$ and $\alg^{R}$ more easily in our proofs, as we produce a crash-stop run from a crash-recovery run when proving properties of the wrapper. The content
of $p$'s buffers changes as we progress through the states $s_i$ of $\unfold(p, s, \rmsgs)$, as the wrapper 
routes the messages to $\alg$ and receives new ones from it.  
The failure detector output (recorded in the labels $l_i$) remains constant through the unfolding: all processes from whom no message was received in the crash-recovery step are suspected.
Finally, the set of failed processes in each label is defined to be empty. We emulate the process recovery that is possible in the crash-recovery model by crash-stop runs in which no processes fail. 

\noindent
Finally, given a crash-stop algorithm $\alg$ with an unreliable failure
detector and with the per-process state space
$\Sigma_p$, we define its \emph{crash-recovery version} $\cralg$
where:
\begin{itemize}
\item the initial states of $\cralg$ constitute the set of crash-recovery configurations $c$ such that there exists a crash-stop configuration $c_s \in \alg.I$ satisfying the following: (i) the initial states of $c$ and $c_s$ correspond to each other and in $c$ all buffers are empty, (ii) no messages are acknowledged,  and (iii) the failed processes
  in $c$ and $c_s$ are the same.

\item the next state of a process $p$ is computed by unfolding, based on the messages $p$ received in this round. 
  
\item the message that $p$ sends to a process $q$ pairs the first element of $p$'s (LIFO) buffer for $q$ with the acknowledgment for the last message that $p$ received from $q$.

\item the execution is short-circuited as soon as a process decides. This is 
achieved by broadcasting a message to all other processes, announcing
that it has decided. When processes receives such a message,
they immediately decide and short-circuits their execution.
\end{itemize}

\noindent
Short-circuiting behavior is a common pattern for consensus
algorithms~\cite{Chandra:UFD,GUERRAOUI:2003}. It can be applied in a black-box 
way and it is sound for any crash-stop consensus algorithm.

\ifthenelse{\boolean{short}}{
A more formal description of the wrapper can be found in the companion  report~\cite{TR}.}
{A more formal description of the wrapper can be found in Appendix~\ref{appendix:wrapper}.}
Given a run of crash-recovery version of an algorithm $\alg$, we define its 
\emph{state trace}, 
$\str(\cdot)$,
as the sequence of configurations obtained by 
removing the labels, projecting each configuration $c$ onto $c.s$, and projecting each local state $c.s_p$ onto $c.s_p.st$.
We overload the function symbol $\str(\cdot)$ to work on both crash-stop and
crash-recovery runs. Note that both crash-stop and crash-recovery state properties are sequences of states from the same state space $\sspace$.

\section{Preservation Results}
\label{sec:preservation-results}

As our first main result, we show that crash-recovery versions of
algorithms produced by our wrapper preserve a wide class of safety
properties.  The class includes the safety properties of consensus:
validity, integrity and agreement
(Section~\ref{sec:system-models}). In other words, if a trace of a
crash-recovery version of an algorithm violates a property, then some
crash-stop trace of the same algorithm also violates that property. We
show this in the non-probabilistic crash-recovery model. However, the
result also translates to the probabilistic model, since all allowed
traces of the probabilistic model are also traces of the
non-probabilistic one.

Preserving all safety properties for all algorithms and failure detectors would be too strong of a requirement, for two reasons. First, as our crash-recovery model assumes nothing about link or process reliability, in finite runs we can give no guarantees about the
accuracy of the simulated failure detectors. Second, the crash-recovery model is synchronous, meaning that different
  processes take steps simultaneously. This is impossible in the
  crash-stop model, which is asynchronous. Thus, the following simple safety
  property $\mathit{OneChanges}$ defined by ``the local state of at most one process changes between
  two successive states in a trace'' holds in the crash-stop model,
  but not in the crash-recovery model (equivalently, we can find a
  crash-recovery trace, but not a crash-stop trace that violates the
  property).

We work around the first problem by assuming that the crash-stop algorithms 
use unreliable failure detectors. For the second problem, we 
restrict the class of safety properties that we wish to
preserve as follows. 
Consider a property $P$. Let $u\notin P$ and $w$ be runs with the same initial states (i.e., $u(1) = w(1)$) such that $u$ is a subword of $w$ (recall we define subwords earlier). $P$ then belongs to the class of properties \emph{not repairable through detours} if $w \notin P$
for all such $u$ and $w$. Intuitively, this means that the sequence of states represented by $u$ inherently violates $P$; so adding ``detours'' by the means of additional intermediate states (forming $w$) does not help satisfy $P$.

%a word$u \notin P$ and $u$ is a subword of $w$ such that $u(1) = w(1)$, i.e., representing runs with the same initial conditions, then also $w \notin P$. 
The property $\mathit{OneChanges}$ is an example of a safety property
that \emph{is} repairable through detours: we can take \emph{any} word that violates $\mathit{OneChanges}$
and extend it to a word that does not violate $\mathit{OneChanges}$.
However, we can easily show that the following safety properties are not repairable through detours:
\begin{itemize}
\item The safety properties of consensus. E.g., consider the
  validity property: given a word
  $u$ such that $v = u(i, p).\dec$ is a non-initial and a non-$\bot$
  value, adding further states between the initial state and $u(i)$
  does not change the fact that $v$ is neither initial nor $\bot$.
  
\item \emph{State invariant} properties, defined by a set $S$ of ``good'' states,
  such that for a trace $u$, $u \in \inv(S)$ only if
  $\forall i.\ u(i) \in S$. Equivalently, these properties rule out
  traces which reach the ``bad'' states in the complement of $S$.
  Intuitively, if a bad state is reached, we cannot fix it
  by adding more states before or after the bad state.
\end{itemize}

\noindent We establish the following lemma, which is essential to the ensuing preservation theorem.

\begin{lemma}[Crash-recovery traces have crash-stop superwords]
  \label{lemma:cr-have-cs-superwords}
  Let $\alg$ be a crash-stop algorithm with an unreliable failure
  detector. Let $\rho^{R}$ be a finite run of the crash-recovery
	wrapper $\cralg$. Then, there
  exists a finite run $\rho$ of $\alg$ such that $\str(\rho^{R})$ is a
  subword of $\str(\rho)$, $\str(\rho^{R})(1) = \str(\rho)(1)$,
  $\last(\str(\rho)) = \last(\str(\rho^R))$, no processes fail in
  $\rho$, and in-flight messages of
  $\last(\rho)$ match the messages in the buffers of $\last(\rho^R)$. 
\end{lemma}

\begin{theorem}[Preservation of detour-irreparable safety properties]\label{theorem: preservation detour irrepairable}
  Let $\alg$ be a crash-stop algorithm with an unreliable failure
  detector, and let $P$ be a safety state property that is not
  repairable through detours. If $\alg$ satisfies $P$, then
  so does $\cralg$.
\end{theorem}

\begin{proof}
  We prove the theorem's statement by proving its contrapositive. Assume
  $\cralg$ violates $P$. By the definition of safety properties~\cite{alpern_defining_1985}, there
  exists a finite run $\rho^R$ of $\cralg$ such that no continuation
  of $\str(\rho^R)$ is in $P$. By
  Lemma~\ref{lemma:cr-have-cs-superwords}, there exists a run $\rho$
  of $\alg$ such that $\str(\rho^R)$ is a subword of $\str(\rho)$.
  By Lemma~\ref{lemma:reliable-irrelevant-prefix}, $\rho$ can be
  extended to an infinite run $\rho \cdot w$ of $\alg$.
  By the choice of $\rho^R$, we have that $\str(\rho^R) \cdot \str(w) \notin P$.
  As $\str(\rho^R) \cdot \str(w)$ is a
  subword of $\str(\rho \cdot w)$, and since 
  $P$ is not detour repairable, then also $\str(\rho \cdot w) \notin P$.
  Thus, $\alg$ also violates $P$.
\end{proof}

Proving that the safety properties of consensus are
detour-irrepairable, means that these properties are preserved by
our wrapper. Since state invariants are also detour-irreparable, they
too are preserved by our wrapper. This makes our wrapper potentially
useful for reusing other kinds of crash-stop algorithms in a
crash-recovery setting, not just the consensus ones.

\begin{corollary}\label{cor:safety_properties}
  If $\alg$ satisfies the safety properties of consensus, then
  so does $\cralg$.
\end{corollary}
\begin{corollary}
  \label{lemma:cs-cr-invariants}
  If $\alg$ satisfies a state invariant, then so does $\cralg$.
\end{corollary}

\section{Probabilistic Termination}
\label{sec:ct-prob-term}

Termination of consensus algorithms depends on stable periods,
during which communication is reliable and no crashes or recoveries occur.
In this section, we first state a general result about so-called
\emph{selective stable periods} for our 
probabilistic crash-recovery model. We then define a generic class of
crash-stop consensus algorithms, which we call \textit{bounded
  algorithms}. We prove that termination for these
algorithms is guaranteed in our probabilistic model when run under the
wrapper. Namely, we prove that, with probability 1, all processes
eventually decide. We also show that the class of bounded algorithms
covers a wide spectrum of existing algorithms including the celebrated
Chandra-Toueg~\cite{Chandra:UFD} and the instantiation of the indulgent framework of~\cite{GUERRAOUI:2003}
that uses failure detectors.

\subsection{Selective Stable Periods}\label{sec: selective stable period}

Similar to~\cite{Dolev:1997}, our proofs will rely on forming
strongly-connected communication components between particular sets of
processes. However, we will require their existence only for bounded
periods of time, which we call selective stable periods.

\begin{definition}[Selective stable period]
  Fix a crash-recovery algorithm $\cralg$. A \emph{selective-stable
    period} of $\cralg$ of length $\Delta$ for a crash-recovery configuration
  $c$ and a set of processes $C$, written $\sper(\cralg, \Delta, c, C)$, is the
  set of all sequences
  $c = c_0, c_1, \ldots c_{\Delta +1}$ of crash-recovery configurations such
  that $\forall i.\ 1 \leq i \leq \Delta+1$ we have $c_i.F = \Pi \setminus C$ and there
  exist $\rmsgs_i$ such that $(c_i, (\rmsgs_i, \emptyset), c_{i+1})$ is a step of $\cralg$
  and $\rmsgs_i(p)(q)$ is defined $\forall p,q \in C$.
%  \D{check if necc. to lose all message from processes not in $P_{\mathbb{R}}(t)$.}
\end{definition}

% We note that, due to the determinism of the algorithms, $c_2, \ldots, c_{\Delta + 1}$  are uniquely defined given $c_1$.

Such selective stable periods must occur in runs of a crash-recovery algorithm
$\cralg$.

\begin{lemma}[Selective stable periods are mandatory]
  \label{lemma:selective-stable-mandatory}
  Fix a crash-recovery algorithm $\cralg$, a positive
  integer $\Delta$ and a selection function $\mathit{sel} : \crcfgs \rightarrow 2^\Pi$,
  mapping crash-recovery configurations to process sets. Then, the set of
  crash-recovery runs
  \[
    \{ c_0, l_0, c_1, \ldots \mid \forall i \geq 0.\
      c_i, l_i,$ $\ldots,  c_{i + \Delta + 1}  \notin \sper(\cralg, \Delta, c_i, \mathit{sel}(c_i)) 
    \}, \text{has a probability of $0$.}
  \]
  
\end{lemma}

The next section shows how our wrapper exploits such periods to construct 
crash-recovery algorithms from existing crash-stop ones in a blackbox 
manner. For future work it might also be interesting to devise  consensus algorithm directly on top of this property.

\subsection{Bounded Algorithms}

We next define the class of \emph{bounded} crash-stop algorithms for which our 
wrapper guarantees termination in the crash-recovery setting. This class 
comprises algorithms which operate in rounds, with an upper bound on the number 
of messages exchanged per round as well the number of rounds correct processes 
can be apart. More formally, they are defined as follows.
\begin{definition}[Bounded algorithms]
  \label{def: bounded alg}
  A crash-stop consensus algorithm (using reliable links and a failure
  detector~\cite{Chandra:UFD}) is said to be \emph{bounded} if it
  satisfies all properties below:

\begin{enumerate}[label=(B\arabic*),leftmargin=*]
\item \label{item:cc-rounds}
  
  \textbf{Communication-closed rounds:} processes operate in
  rounds. The rounds must be
  communication-closed~\cite{elrad_decomposition_1982}: only the
  messages from the current round are considered.
  
\item \label{item:external-change} \textbf{Externally triggered state changes:} After the first
  step of every round, the processes change state only upon receipt of
  round messages, or on a change of the failure detector output.
  
\item \label{item:bounded-messages} \textbf{Bounded round messages:} There exists a bound $B_s$ such
  that, in any round, a process sends at most $B_s$ messages to any
  other process. 

\item \label{item:round-gap} \textbf{Bounded round gap:} Let $N_c = N - N_f$, that is, the
  number of correct processes according to the algorithm's resilience
  criterion. Then, there exists a bound $B_\Delta$, such that the
  fastest  $N_c$ processes are at
  most $B_\Delta$ rounds apart. %(measured by round number). Hence, in any reachable configuration
  %$c$, at least $N_c$ processes have a round number greater or equal
  %to $r_\mathit{max}(c) - B_\Delta$, with $r_\mathit{max}(c)$ the
  %highest round of all processes in $c$. %%use in a proof?

\item \label{item:bounded-term} \textbf{Bounded termination:} There
  exists a bound $B_\mathit{adv}$ such that for any reachable
  configuration $c$ where any $N_c$ fastest processes in $c$ are
  correct, the other processes are faulty, and the failure detector
  output at these processes is perfect after $c$, then all of these
  $N_c$ processes decide before any of them reaches the round
  $r_\mathit{\max}(c) + B_\mathit{adv}$.
\end{enumerate}
\end{definition}
We can check that an algorithm satisfies property (i) B4 by checking under what condition(s) a process increments its round number, and (ii) B5 by observing the algorithm's termination under perfect failure detection given a quorum of correct processes. Section~\ref{sec:exampl-bound-algor} shows an example of how to check that properties B4 and B5 hold. 
Using this definition, we can prove that bounded algorithms terminate in 
selective stable periods of sufficient length.

\begin{theorem}[Bounded selective stable period termination]
  \label{theorem: bounded selective-stable-period}
  Let $\alg$ be a bounded algorithm and $\cralg$ its wrapped crash-recovery
  version.
  Let $c$ be a reachable crash-recovery configuration of the $\cralg$, and
  let $C$ be some set of $N_c$ fastest processes in $c$. Then, there
  exists a bound $B$, such that, for any selective stable period of length
  $B$ for $c$ and $C$, all processes in $C$ decide in $\cralg$. Moreover, the
  bound $B$ is independent of the configuration $c$.
\end{theorem}
\begin{proof}
  We partition the processes from $C$ into the set $A$ (initially $C$)
  and $\mathit{NA}$ (initially~$\emptyset$). The processes in $A$ will advance
  their rounds further, and the processes in $\mathit{NA}$ will not
  advance, but will have already decided. Let $p$ be some slowest
  process from $A$ in $c$. We first claim that $p$ advances or decides
  in any selective stable period for $c$ and $C$ of length
  at most $B_\mathit{slow}$, defined as $B_\mathit{slow} = B_s \cdot B_\Delta + 1$. Denote the period
  configurations by $c = c_0, c_1, \ldots$. Then, using
  Lemma~\ref{lemma:cr-have-cs-superwords}, we obtain a crash-stop
  configuration $c^s_1$ from $c_1$ that satisfies the conditions of
  bounded termination, with the processes from $C$ being correct, and
  the others faulty. We consider two cases.

  First, if $p$ advances in
  the crash-stop model after receiving all the in-flight round
  messages in $c^s_1$ from other processes in $A$,
  % with the messages ordered by the sender
  then it also advances in the crash-recovery
  model after receiving these messages. Moreover, our wrapper delivers
  all such messages within $B_s \cdot B_\Delta$ steps, as
  \begin{enuminline}
  \item it uses LIFO buffers;
  \item bounds $B_\Delta$ and $B_s$ apply to $c^s_1$
    by~\ref{item:bounded-messages} and~\ref{item:round-gap}; and
  \item by Lemma~\ref{lemma:cr-have-cs-superwords}, the same bounds
  also apply to $c_1$.
  \end{enuminline}

  Second, if $p$ does not advance in the
  crash-stop model, then, since the failure detector output remains
  stable, and since no further round messages will be delivered to $p$
  in the crash-stop model, the requirement~\ref{item:external-change}
  ensures that $p$ will not advance further in the crash-stop setting.
  Moreover, requirements~\ref{item:bounded-term}
  and~\ref{item:external-change} ensure that $p$ must decide after
  receiving all of its round messages; we move $p$ to the set
  $\mathit{NA}$.

  We have thus established that the slowest process from $A$ can move
  to either $\mathit{NA}$ or advance its round after $B_\mathit{slow}$
  steps. Next, we claim that we can repeat this procedure by picking
  the slowest member of $A$ again. This is because the
  procedure ensures
  that the processes in $\mathit{NA}$ always have round numbers lower than the
  processes in $A$. Thus, due to~\ref{item:cc-rounds}
  and~\ref{item:external-change}, the processes in $\mathit{A}$ cannot
  rely on those from $\mathit{NA}$ for changing their state.

  Lastly, we note that this procedure needs to be repeated at most
  $B_\mathit{iter} = N \cdot (B_\Delta + B_\mathit{adv})$ times before all processes move to
  the round $r_\mathit{max}(c) + B_\mathit{adv}$, by which
  point~\ref{item:bounded-term} guarantees that all processes
  terminate. Thus $B_\mathit{slow} \cdot B_\mathit{iter}$ gives us
  the required bound $B$.
\end{proof}

The main result of this paper shows that the wrapper guarantees all consensus
properties for wrapped bounded algorithms, including termination.  

\begin{theorem}[CR consensus preservation]
  \label{corollary: CR termination
    of bounded alg} If a bounded algorithm $\alg$ solves consensus in
  the crash-stop setting, then $\cralg$ also solves consensus in the
  probabilistic crash-recovery setting.
\end{theorem}
\begin{proof}
  By Corollary~\ref{cor:safety_properties}, we conclude that $\cralg$
  solves the safety properties of consensus in the crash-recovery
  setting. For (probabilistic) termination, the result follows from
  Theorem~\ref{theorem: bounded selective-stable-period} and
  Lemma~\ref{lemma:selective-stable-mandatory}, using as the selection
  function for Lemma~\ref{lemma:selective-stable-mandatory}
  \begin{enuminline}
  \item any function that selects some $N_c$ fastest processes in a
    configuration if no process has decided yet and
  
  \item all processes, if some process has decided.
  \end{enuminline}
  The latter allows us to propagate the decision to all processes,
  due to short-circuiting in  $\cralg$.
\end{proof}

%\begin{lemma}\label{lemma: CT bounded}
%Chandra-Toueg's algorithm~\cite{Chandra:UFD} is a bounded algorithm.\end{lemma}

\subsection{Examples of Bounded Algorithms}
\label{sec:exampl-bound-algor}

We next give two prominent examples of bounded algorithms: the
Chandra-Toueg (CT) algorithm~\cite{Chandra:UFD} and the instantiation of the indulgent framework of~\cite{GUERRAOUI:2003}
that uses failure detectors.
For these algorithms, rounds are composite, and consist of the
combination of what the authors refer to as rounds and
phases. Checking that the algorithms then satisfy
conditions~\ref{item:cc-rounds}--\ref{item:bounded-messages} is
straightforward. 
%Indeed, these checks could easily be implemented syntactically. 
\ref{item:round-gap} holds for
the CT algorithm with $B_\Delta = 4 \cdot N$: take the
fastest process $p$ in a crash-stop configuration; if its CT round number
$r_p$ is $N$ or less, the claim is immediate. Otherwise, $p$ must have
previously moved out of the phase 2 of the last round $r_p'$ in which
it was the coordinator, which implies that at least $N_c$ processes
have also already executed $r_p$. Since CT uses the
rotating coordinator paradigm, $r_p - r_p' \le N$; as each round
consists of $4$ phases, $B_\Delta = 4 \cdot N$. For the algorithm
from~\cite{GUERRAOUI:2003}, processes only advance to the next round
(which consists of two phases) when they receive messages from $N_c$
other processes. Thus, $B_\Delta = 2$. Finally, proving the
requirement~\ref{item:bounded-term} is similar to, but simpler than
the original termination proofs for the algorithms, since it only
requires termination under conditions which includes perfect failure
detector output. For space reasons, we do not provide the full proofs
here, but we note that
$B_\mathit{adv} = \mathit{phases} \cdot \floor{{N}/{2}}$ for both
algorithms, where $\mathit{phases}$ is the number of phases per
algorithm round. Intuitively, within this many rounds the execution
hits a round where a correct processor is the coordinator; since
we assume perfect failure detection for this period, no process
will suspect this coordinator, and thus no process will move out
of this round without deciding.
\vspace{-5pt}\begin{corollary}
The wrapped versions of Chandra-Toueg's algorithm~\cite{Chandra:UFD} and the 
instantiation of the indulgent framework of~\cite{GUERRAOUI:2003} using failure 
detectors solve consensus in the probabilistic crash-recovery model.
\end{corollary}

%How do we fix this? We could relax this by saying that the only messages received are the ones sent after $c$; or perhaps we allow the processes to receive at most 1 message sent before $c$. We could implement this using LIFO buffers.\\
\vspace{-10pt}\section{Concluding Remarks}
\label{sec: conclusion}
\vspace{-5pt}This paper introduced new system models that closely capture the messy reality
of distributed systems. Unlike the usual distributed computing models,
we account for failure and recovery patterns of processes and links in
a probabilistic and temporary, rather than a deterministic and perpetual
manner.  Our models allow an unbounded number of processes and
communication links to probabilistically fail and recover, potentially
for an infinite number of times. We showed how and under what
conditions we can reuse existing crash-stop distributed algorithms in
our crash-recovery systems. We presented a wrapper that allows
crash-stop algorithms to be deployed unchanged in our crash-recovery
models. The wrapper preserves the correctness of a wide class of
consensus algorithms.

Our work opens several new directions for future
investigations. First, we currently model failures of processes as
well as communication links individually and independently, with a
non-zero probability of failing/recovering at any point in time. 
\ifthenelse{\boolean{short}}{In the companion technical report~\cite{TR}}{In
Appendix~\ref{sec:nastier-environment}}, we sketch how our results can
be extended to systems where some processes may never even recover
from failure. 
It is interesting to investigate what results can be
established with more complicated probability distributions, e.g., if the model is weakened to allow processes and links to
fail/recover on average with some non-zero
probability~\cite{Fetzer}. Second, our wrapper fully persists the
processes state. Studying how to minimize the amount of persisted
state while still allowing our results (or similar ones) to hold is
another promising direction. Finally, we focus on algorithms
that depend on the reliability of message delivery. Some
algorithms, notably Paxos~\cite{lamport_part-time_1998}, do not.
Finding a modular link abstraction for the crash-stop setting that
identifies these algorithms is another interesting topic. For those algorithms, we speculate that preserving termination in the crash-recovery model
is simpler. 
%\vspace{-10pt}

% \begin{theorem}[CT termination]
%   \label{thm:ct-termination}
%   $CT^R$ terminates with probability $1$ in the probabilistic
%   crash-recovery setting.
% \end{theorem}
% \begin{proof}
%   Consequence of Lemmas~\ref{lemma:selective-stable-mandatory}
%   and~\ref{lemma:ct-cr-bound}.
% \end{proof}

\ifthenelse{\boolean{short}}{}{

%%%%%%%%%%%%%%%%%%%%%%%%%%%%%%%%%%%%%%%%%%%%%%%%%%%%%%%%%%%%%%%%%%%%%%%%%%%%%%%%
\appendix
%%%%%%%%%%%%%%%%%%%%%%%%%%%%%%%%%%%%%%%%%%%%%%%%%%%%%%%%%%%%%%%%%%%%%%%%%%%%%%%%

\section{Wrapper Description}\label{appendix:wrapper}
  
In this section we provide the formal description of the wrapper that produces the crash-recovery version $\cralg$ of a crash-stop algorithm $\alg$. The message and local state spaces are denoted by $\msgs^R$ and $\Sigma_p^R$ respectively.

The message space of $\cralg$ is defined as $\msgs^{R} = (\msgs \cup \set{\hbeat}) \times (\msgs \cup\set{m_\emptyset})$. $\msgs$ is the message set of $\alg$ which we extend with a special heartbeat message $\hbeat$ being sent when no message needs to be sent. The second part of each message represents an acknowledgment, confirming the receipt of the last message on the channel in the opposite direction. The $m_\emptyset$ message indicates that no messages (including acknowledgment messages)  have been received so far.

The types of the components $(st,\buff,\acks)$ of a record $s\in\Sigma_p^R$ are
\begin{itemize}
\item $st$ of type $\Sigma_p$ storing the state of $p$ in the target crash-stop algorithm.
\item $\buff$ of type $\procs \to \msgs^*$, where we recall that $\msgs^*$ stands for the set of finite sequences of messages from $\msgs$. These are $p$’s outgoing message buffers, with one buffer for each process in the system (including one for $p$).
\item $\acks$ of type $\procs \to (\msgs \cup\set{m_\emptyset})$. This records, for each process $q$, the last message that $p$ received from $q$ (if any). This will be used for acknowledgments.
\end{itemize}

Given a crash-recovery state $s$, a process $p$, and the
partial function $\rmsgs : \procs \rightarrow \msgs^{R}$ of
messages received by $p$ in the given round, we define the
\emph{p's local step unfolding for $s$ and $\rmsgs$}, written
$\unfold(p, s, \rmsgs)$ as follows.
First, let:

\begin{itemize}
\item $m_q$ be $\rmsgs(q)$ if $q \in dom(\rmsgs)$, and let $m_q$ be
  $\bot$ otherwise. That is, $m_q$ is the message from $q$ that $p$
  receives, using $\bot$ to indicate that no message was received.

\item Let $\fdo = \set{q \mid m_q = \bot}$. That is, all processes
  from whom no message was received in this step are suspected.
\end{itemize}

Then, unless $p$ has decided yet and broadcasts the decision, the following sequence $\unfold(p, s, \rmsgs)$ of intermediate steps is taken. I.e., $\unfold(p, s, \rmsgs)$  defines a sequence of 
crash-recovery states and crash-stop labels
$s_0, l_0, s_1, l_1, s_2, \ldots s_n$, where the intermediate state $s_i$
represents the state record of $p$ after processing the message 
of the $i$-th process, defined as follows:
  \begin{enumerate}
  \item $s_0 = s_p$
    
  \item Recalling that we number the processes from $1$ to $N$,
    $l_i$ and $s_{i+1}$ are computed as follows, for $0 \le i < N$:
    \begin{enumerate}
    \item Unpack the message from $i+1$ if one has been received, Let
      $(m,a) = m_{i+1}$ if $m_{i+1} \neq \bot$, and let $m=\bot$ and
      $a=\bot$ otherwise. If $m\neq hb$ and $m\neq\bot$ we check that
      the message has not been acknowledged yet, $s_i.ack(i+1)\neq m$,
      in which case we feed the message to $\alg$'s next function:
      $s_{i+1}.st = \alg.\nnext(s_i.st, (i+1,m), \fdo)$. We also need
      the next state if \emph{no} message (or a duplicate) has been
      received, as $(s_i, l_i s_{i+1})$ needs to be a transition of
      the crash-stop system.  Hence, if $m = \bot$ or $s_i.ack(i+1) = m$, we feed
      $\bot$ in and not $m$, i.e.,
      $s_{i+1}.st = \alg.\nnext(s_i.st, \bot, \fdo)$.
			% In case $p$ has decided in $\alg$ in this step, the unfolding stops and the decision is broadcast to all other processes from then on, all buffers are emptied and $s_N.\buff(j)$ is set to the message $\alg$ sends 
			% to announce a decision for all $j$.
			% \yap{hmmm, don't we have a problem with communication closed rounds?}

    \item we set $l_i = (p, (i+1, m), \emptyset, \fdo)$ or
      $l_1 = (p, \bot, \emptyset, \fdo)$ accordingly. In both cases, the set of
      failed processes in this label is empty.

    \item We remove the acknowledged message (if any) from the head of the outgoing
      buffer for this process. More precisely, let $b'$ be the buffer
      obtained as follows. 
			First, copy the buffer of all other processes except $i+1$.\\
      $b'(j) = s_i.\buff(j)$ for all $j \neq i + 1$. 
			
			Next, if there are messages to send to process $i+1$, i.e., 
      $s_i.\buff(i+1) \neq []$, and if $a\neq \bot$ and $a$ is equal to $\hd(s_i.\buff(i+1))$, then we let $b'(i+1)$ be $\tl(s_i.\buff(i+1))$; otherwise,
      let $b'(i+1) = s_i.\buff(i+1)$. 
			
			To add a potential new message to the buffer, let
      $\mathit{new} = \alg.\send(s_{i+1}.st)$. These are the new
      messages that $p$ wishes to send, at most one destined to each process in the system. We define
      $s_{i+1}.\buff(j) = b'(j)$, if $\mathit{new}(j)$ is undefined,
      and $s_{i+1}.\buff(j) = \mathit{new}(j)  \cdot b'(j) $ otherwise. Notice that we add the new message
      at the head of the list; as we will also remove messages from the head
      when in the $\send$ function, our buffers are LIFO.

    \item % Let $a = m'$, if $m_{i+1} = (m', a')$ for some $a'$,
      %and let $m' = s_i.\acks(i+1)$ otherwise. 
			If $m \neq \bot$ and $m\neq \hbeat$, then 
      $s_{i+1}.\acks(i+1) = m$, and $s_{i+j}.\acks(j) = s_i.\acks(j)$
      for $j \neq i + 1$.
    \end{enumerate}
  \end{enumerate}
Finally, given a crash-stop algorithm $\alg$ using an unreliable failure
detector with range $2^\procs$ and with the per-process state space
$\Sigma_p$, we define its \emph{crash-recovery version} $\cralg$
where:
\begin{itemize}
\item $\cralg.I$ is the set of crash-recovery configurations $c$ such
  that there exists a crash-stop configuration $c_s \in \alg.I$ such
  that:
  \begin{enumerate}
  \item for each $p$:
    \begin{enumerate}
    \item $c.s_p.st = c_s.s_p$
    \item $c.s_p.\buff(q) = []$ for all $q$. That is, initially, no messages are
      buffered.
    \item $c.s_p.\acks(q) = m_\emptyset$ for all $q$. Initially, no messages are
      acknowledged.
    \end{enumerate}
  \item $c.F = c_s.F$
  \end{enumerate}

\item $\cralg.\nnext(s_p, \rmsgs) = \last(\unfold(p, s_p, \rmsgs))$.
  
\item $\cralg.\send(s_p)(q)$ is:
  \begin{itemize}
  \item $(\hd(s_p.\buff(q)), s_p.\acks(q))$ if $s_p.\buff(q) \neq []$,
    and
    
  \item 
    $(\hbeat, s_p.\acks(q))$, otherwise, i.e., if we have nothing else to send.
  \end{itemize}
\end{itemize}

\section{Markov chains}
\label{sec:markov-chains}

Our probabilistic crash-recovery model uses Markov
chains~\cite{baier_principles_2008}. We recall the basic notions here
which are relevant for our proofs.

A (discrete-time) \emph{Markov chain} is a tuple
$(S, P, \mcinit)$ where:
\begin{itemize}
\item $S$ is a countable, non-empty set of states
\item $P : S \times S \rightarrow [0, 1]$ is the \emph{transition
    probability function} such that for all states $s$:
  $
    \sum_{s' \in S} P(s, s') = 1.
  $
  
\item $\mcinit : S \rightarrow [0, 1]$ is  the \emph{initial distribution},
  such that $\displaystyle\sum\limits_{s \in S} \mcinit(s) = 1$.
  
\end{itemize}

For a Markov chain, the \emph{cylinder} set spanned by a finite
word $u$ over $S$ is defined as the set:
$
  \mathit{Cyl}(u) = \set{ u \cdot v \mid v \text{ is an infinite word over S}}.
$
These sets serve as the basis events of the $\sigma$-algebras of
Markov chains. If $u = s_0, s_1, \ldots s_k$, then
$
  Pr(\mathit{Cyl}(u)) = \mcinit(s_0) \cdot P(s_0, s_1) \cdot P(s_1, s_2) \cdots P(s_{k-1},
  s_k).
$

\section{Omitted Proofs}\label{app:omitted proofs}

\subsection{Proof of Lemma~\ref{lemma:cr-have-cs-superwords}}
  Informally, we obtain $\rho$ by unfolding the crash-recovery steps
	of $\rho^{R}$. In other words, we extend each crash-recovery step to its corresponding crash-stop steps. Recall that a single step in the crash recovery setting may abstract the receipt of multiple messages, while every received message in the crash-stop setting is step of the algorithm. 
    
 %We rely on the unreliability of the failure detector for the result.\ogi{Where is this visible?} 
 The only difficulty is 
	handling process recovery, which does not exist in the crash-stop system. 
	Here, we exploit the fact that, in asynchronous systems, crashed processes
    are indistinguishable from delayed ones. This enables us to keep
    $\rho$ free from failed processes.

 We now state our claim more formally, and then prove it by induction.  
 We claim that there exists a finite run $\rho$ such that:

  \begin{enumerate}[label=(P\arabic*)]
  \item \label{item:cr-cs-subword} $\str(\rho^{R})$ is a subword of
    $\str(\rho)$, with the same first and last character: 
		$\str(\rho)(1) = \str(\rho^R)(1)$ and
    $\last(\str(\rho)) = \last(\str(\rho^R))$. 

  \item \label{item:cr-cs-failures} No failures occur in $\alg$, $\rho(i).F = \emptyset$ for all 
	$1 \le i \le \len{\rho}$.%,
   %and such that
   
 \item \label{item:cr-cs-messages} In-flight messages of $\alg$
		correspond to messages in the buffers of $\alg^R$,\\ 
	 $\last(\rho).M = \{(q, p, m) \mid m \in \last(\rho^{R}).s_p.\buff(q) 
	 \land \last(\rho^R).s_q.\acks(p) \neq m\}$.  
	 Here, we overload $\in$, writing $m \in \buff$ even though $\buff$ is a 
	 sequence.
  \end{enumerate}

 \noindent Proving the base case is easy, as $\rho^R(1).s.st$ is an initial state of
  $\alg$, and both the in-flight set of messages $M$ and all the buffers are empty for initial
  configurations (follows from our wrapper's definition).

  We now prove the step case. Let $\rho^R$ and $\rho$ be as above, let:
  \begin{itemize}
  \item $c = \last(\rho)$,
  \item $c^R = \last(\rho^R)$,
  \item $(c^R, l^R, c^{{\prime}R}) \in \crtr(\alg^R)$, and
  \item $\rho^{{\prime}R}$ be the extension of $\rho^R$ with
    $(c^R, l^R, c^{{\prime}R})$.
  \end{itemize}
  In other words, $c'^R$, is a next ``letter'' we could add to a valid run of $\alg^R$.
  Thus, we will prove that we can extend $\rho$ with a series of
  transitions, to obtain a run $\rho'$ such that the claim holds for
  $\rho^{{\prime}R}$ and $\rho'$. The extension is obtained by concatenating
  the unfoldings of each process's local step. More precisely, we
  define the sequences (fragments) $\frag_0 \ldots \frag_N$ where $\frag_i$ is the unfolding of process $i$'s local steps:
  \begin{itemize}
  \item $\frag_0 = [c]$
    
  \item if $i \in c^R.F$, then $\frag_i = [\last(\frag_{i-1})]$
  \item if $i \notin c^R.F$, then $\frag_i$ is the sequence
    $[c_0, l_0, c_1, \ldots c_N]$ such that:
    \begin{itemize}
    \item $c_0 = \last(\frag_{i-1})$

    \item
      $c_0.s_i,\, l_0,\, c_1.s_i,\, l_1, \ldots
      c_N.s_i$ is $\unfold(i,\, c_0.s_i,\, l^R.\rmsgs)$

    \item $c_k.s_j = c_{k-1}.s_j$ for $j \neq i$ and
      $0 < k \le N$ (the state of other processes does not change)

    \item $c_k.M = c_{k-1}.M \setminus \set{l_{k-1}.\rmsg}
      \cup \set{(j, i, m) \mid (j, m) \in \graph{\send(c_{k-1}.s_i)}}$,
      for $0 < k \le N$ (messages delivered to $i$ are removed the in-flight message $M$ and new messages sent by $i$ are added)

    \item $c_k.F = \emptyset$ for $0 < k \le N$
      
    \end{itemize}

  \end{itemize}
  
\noindent  Since the successive fragments overlap at one state, we define
  $$
    \rho' = \rho \cdot \tl(\frag_0) \cdot \tl (\frag_1) \cdot \ldots
    \cdot \tl(\frag_N).
  $$
  and we let $c' = \last(\rho')$. It is easy to check that 
  the definition of $\unfold$ and the correspondence we establish
  between the buffers ensure that $\rho'$ is a (finite) run of $\alg$.
  Finally, we show the claim:
  \begin{itemize}
  \item For all $p$, 
    if $p \notin c^R.F$ then:
    \begin{align*}
     c^{{\prime}R}.s_p &= \cralg.\nnext(c^R.s_p, l^R.\rmsgs(p))  \tag{def. of CR-steps}\\
                       &= \last(\unfold(p, c^R.s_p, l^R.\rmsgs(p)))  \tag{def. of wrapper} \\
                       &= c'.s_p \tag{def. of $c'$}
    \end{align*}
    and if $p \in c^R.F$, then
    \begin{align*}
      c^{{\prime}R}.s_p &= c^R.s_p \tag{def. of CR-steps} \\
                        &= c.s_p \tag{inductive hypothesis} \\
                        &= c'.s_p \tag{def. of $c'$}
    \end{align*}

    Hence, $c^{{\prime}R}.s.st = c'.s$, and given the inductive
    hypothesis, we conclude that $\str(\rho^{{\prime}R})$ is a subword
    of $\str(\rho')$, and \ref{item:cr-cs-subword} holds.

  \item $\rho'(i).F = \emptyset$ for all $1 \le i \le \len{\rho'}$.
    Thus, \ref{item:cr-cs-failures} holds.
  \end{itemize}
  
%\yap{only for finite runs, but this is enough, since we show termination later...}

	\subsection{Proof of Corollary~\ref{cor:safety_properties}}

We first prove that each of the safety properties of consensus is a detour-irrepairable safety property. The rest follows directly from Theorem~\ref{theorem: preservation detour irrepairable}.
\begin{enumerate}
\item Validity: Consider a word $w$ such that
  $w(i, p).\dec = x$, where
  $ x\neq~\bot ~\land~ x\neq w(1, q).inp,~ \forall q$. In other words, $w$ violates validity. %over $\alg$
%, we show in what follows that adding intermediate steps does not undo a violation in the validity property.
  %\ogi{Actually, we don't need to care about the wrapper here at all -- we just need to show that adding intermediate states to $w$ doesn't ``fix'' $w$ with respect to integrity.} 
    %In other words, adding more states between $w(1)$ and $w(i')$does not eliminate the fact that $\exists i'\geq i: w'(i', p).\dec = x$ as well. Thus, adding further states does not change the fact that $x$ is neither $\bot$ nor equal to an initial value.
    Adding states between the initial state, $w(1)$, and $w(i)$, creates a new word where $w'(1,q).inp = w(1, q).inp,~ \forall q$ and $\exists i'\geq i: w'(i', p).\dec = x$. Thus, adding further states does not change the fact that $x$ is neither $\bot$ nor equal to an initial value. Validity is hence detour-irreparable.
    
    %In other words, $\forall~w'$ over $\cralg$ we have $w'(1,q).inp = w(1, q).inp,~ \forall q$. Since our wrapper may only add further states between the initial state, $w(1)$, and $w(i)$, e.g., due to unfolding, then $\exists i'\geq i: w'(i', p).\dec = x$ as well. Thus, adding further states does not change the fact that $x$ is neither $\bot$ nor equal to an initial value.
    
\item Integrity: Consider a word $w$ such that $w(i, p).\dec = x$ and $w(i', p).\dec = x'$, where $x\neq\bot$, $i' > i$  and $x'\neq x$. Adding further states between $w(1),~w(i)$ and $w(i')$ results in word $w'$ where $\exists j'\geq j\geq i: w'(j, p).\dec = x\land w'(j', p).\dec = x'$ and $w'(1,q).inp = w(1, q).inp,~ \forall q$. Thus, adding further states does not change the fact that $x$ is neither $\bot$ nor that $x'\neq x$ and hence integrity is detour-irreparable.

%\D{Our current model does not capture the act of invoking a decision, for example, we cannot express a process that keeps deciding the same value, say $x$, repeatedly.}

\item Agreement:  Consider a word $w$ such that $w(i, p).\dec = x$ and $w(i', q).\dec = x'$, where $x,x'\neq\bot$ and $x'\neq x$. Adding further states between $w(1),~w(i)$ and $w(i')$ creates a new word $w'$ where $\exists j'\geq j\geq i: w'(j, p).\dec = x\land w'(j', q).\dec = x'$. Notice that even if $\exists j'\geq j''\geq j: w'(j'', p).\dec = \bot$, this does not change the fact that $w'(j, p).\dec = x$. Thus, adding further states does not change the fact that two different processes $p$ and $q$ have decided $x\neq \bot$ and $x' \neq \bot$ $x$ respectively, where $x\neq x'$. This shows that agreement is detour-irreparable.
\end{enumerate}

\subsection{Proof of Lemma~\ref{lemma:selective-stable-mandatory}}

We first prove two auxiliary results.

\begin{lemma}[Selective recovery always possible]
  \label{lemma:full-recovery}
  For any crash-recovery algorithm $\cralg$ and any configuration $c$
  of its probabilistic crash-recovery system, and $C$ any set of processes
  $$
    \left(\sum\limits_{c'} P(c, c') \cdot \truthof{c'.F = \Pi \setminus C}\right) > \epsfail^{|C|} (1-\epsfail)^{N-|C|}.
  $$
\end{lemma}
\begin{proof}
The normalization factor  $\mathit{nf_\mathit{trans}}(c)$ in the definition of the Markov chains in Section~\ref{sec:probabilistic-cr} is the inverse of:
  \begin{align*}
    \sum\limits_{c'} \sum\limits_\rmsgs & \truthof{\exists \fails, c'.\ (c, (\rmsgs. fails), c') \in \crtr(\cralg)} \cdot \\
    & \prod\limits_{p, q} \Big(\truthof{\rmsgs(p)(q) \text{ defined}}
    \cdot \pnet(q, p, c.n) \\
    & \qquad + \truthof{\rmsgs(p)(q) \text{ undefined}}\cdot\truthof{ q \notin c.F}
    \cdot (1 - \pnet(q, p, c.n)) \Big)
  \end{align*}
  and thus larger than $1$.
Together with the definition of $\mathit{trans}(c)$ and $\epsfail$, we can see that the left hand side is
  larger than $\epsfail^{|C|} (1-\epsfail)^{N-|C|} \cdot \mathit{nf}_\mathit{trans}$. %\mathit{nf}_\mathit{trans}
\end{proof}
 
\begin{lemma}[Selective stable periods possible]
  \label{lemma:selective-stable-possible}
  Consider a crash-recovery algorithm $\cralg$, a positive
  integer $\Delta$, and a set of processes $C$. Then, there exists an $\epsstable > 0$ such that for
  any configuration $c_0$, there exist subsequent states
  $c_1, c_{2}, \ldots c_{B}$ such that $c_1, \ldots c_{B}$ is a
  selective-stable period
  %\ogi{Does not fit the definition of ss-periods; the definition requires $c_0.F = \Pi \setminus P_\mathbb{R}(c_0)$ and this doesn't hold for an arbitrary $c_0$} and $P(c_1, \ldots c_{B})>\epsstable$.
\end{lemma}
\begin{proof}
  From Lemma~\ref{lemma:full-recovery}, we can find a $c_1$ such that
  $P(c_0, c_1) > \epsfail^{|C|} (1-\epsfail)^{N-|C|}$. Note that for $i \ge 1$, given $c_i$,
  $c_{i+1}$ is unique, since the algorithm is deterministic and since
  we require that all messages between processes in $C$ go through ($\rmsgs$ is defined).
  %\ogi{This does not seem to hold any more?}. 
  From the definition of $\mathit{trans}$, we can easily
  deduce $P(c_i, c_{i+1}) > \epsnet^{|C|^2} \cdot \epsfail^{|C|} (1-\epsfail)^{N-|C|}$ for
  $i > 1$, since $c_i.F = \Pi-\{C\}$ by construction. Thus,
  $\epsfail^{|C|} (1-\epsfail)^{N-|C|} \cdot (\epsnet^{|C|^2} \cdot \epsfail^{|C|} (1-\epsfail)^{N-|C|})^\Delta$ can be used as
  $\epsstable$.
\end{proof}

Finally, we prove Lemma~\ref{lemma:selective-stable-mandatory}.
\begin{proof}
  Let $\mathit{noss}(k)$ be the set of all finite runs of $\cralg$ of
  length $k$ which have no selective stable periods.
  Define
  $
    \mathit{Cyl}(\mathit{noss}(k)) = 
    \bigcup\limits_{\tau \in \mathit{noss}(k)} \mathit{Cyl}(\tau).
  $
  By
  Lemma~\ref{lemma:selective-stable-possible}, we get:
  $$
    Pr(\mathit{Cyl}(\mathit{noss}(k))) < (1 - \epsstable)^{\floor{\frac{k}{\Delta}}}.
  $$
  For the set $S$ of
  infinite runs with no stable periods, we have:
  $
    S = \bigcap\limits_k \mathit{Cyl}(\mathit{noss}(k))
  $
  and since
  $
    \mathit{Cyl}(\mathit{noss}(k+1))
    \subseteq 
    \mathit{Cyl}(\mathit{noss}(k))
  $
  we have that
  $
    Pr(S) = \lim\limits_{k \rightarrow \infty} Pr(\mathit{Cyl}(\mathit{noss}(k)))
    \le \lim\limits_{k \rightarrow \infty} (1 - \epsstable)^{\floor{\frac{k}{\Delta}}} = 0,
  $
  since $\epsstable > 0$ and $\Delta$ is positive.
\end{proof}

\section{A nastier probabilistic environment}
%\subsection{A nastier probabilistic environment}
\label{sec:nastier-environment}

In the probabilistic model described in
Section~\ref{sec:probabilistic-cr}, we required the probabilistic
failure pattern $\pfail$ to assign to all processes a positive
probability of being up at all times. 
%This model makes it easy to implement a perfect failure detector, by declaring all processes as ``up'' at all times. Hence, the model could be deemed as unrealistic, since in reality we could expect some processes to crash forever.
This allows us to form arbitrary stable components for bounded
lengths of time. Here, we define a more difficult setting, where
some processes can fail forever. However, the number of the remaining
processes must be large enough to allow termination. We thus
change the condition on the failure pattern from Section~\ref{sec:probabilistic-cr}
as follows. We require that there exists a set $C$ of (intuitively ``correct'') processes such
that:
\begin{itemize}
\item there exists an $\epsfail > 0$ such that
  $\forall p \in C.\, \forall t.\, \pfail(p, t) \in (\epsfail, 1 - \epsfail)$
\item $\card{C} > N - N_f$
  
\item for all processes $p \notin C$, there exist only a finite
  number of values $t$ such that $\pfail(p, t) > 0$.
\end{itemize}

The third restriction allows us to eventually form a bounded stable period
with a strongly connected component of correct processes. Without it, we
could have a process whose probability of being up is oscillating between $0$ and $1$.
Clearly, in a setting where only one process has a probability of $1$, and
all others of $0$ at all times, termination is impossible.

As this definition still does not add any runs to the system that did
not exist in the non-probabilistic environment, all results from
Section~\ref{sec:preservation-results} still hold. We next sketch how our
definition of bounded algorithms and our proofs
need to change to preserve the termination result in this more adversarial
setting. First, note that the specification of termination has to change
in this setting; as processes not from $C$ can crash forever, we cannot
require them to terminate.

Intuitively, consider the point in time in a crash-stop run after which all faulty
processes have already crashed. Faulty processes in this context are
those processes that crash permanently, i.e. not in the set $C$ above.
A corresponding run in the crash-recovery system would be in some
state where the processes from $C$ can recover with
positive probability, with processes being at different rounds. Our system now can
be seen similar to our probabilistic model described in
Section~\ref{sec:probabilistic-cr}, but now started from a
different set of possible initial states. We need to adapt our bounded round gap
property~\ref{item:round-gap} to account for this difference. Namely,
we change it to require the bound to hold for any configuration where
all the faulty processes have crashed, and where we no longer assume that the
fastest processes are correct. The bound is dynamic, in that it can depend
on the given configuration. The bounded termination requirement~\ref{item:bounded-term}
also changes in a similar way; we no longer assume that the fastest processes
are correct, but that all the faulty processes have already crashed. However,
we do not need to make the bound $B_\mathit{adv}$ dynamic.

With these modifications, we can reuse the proof of Theorem~\ref{theorem:
  bounded selective-stable-period} to show that
non-faulty processes of this nastier model decide, and hence
terminate.
}
\end{document}